\documentclass[a4paper,UKenglish]{lipics-v2016}

\usepackage[utf8]{inputenc}
\usepackage{amsmath}
\usepackage{wrapfig}
\usepackage{amssymb}
\usepackage{thmtools}
\usepackage{epsfig}
\usepackage{bbold}
\usepackage{multicol}
\usepackage{multirow}
\usepackage{MnSymbol}
\usepackage{tikz}
\usepackage{pgfplots}
\pgfplotsset{compat=1.12}

\usetikzlibrary{arrows,decorations.pathmorphing,backgrounds,positioning,fit,calc,automata}
\usetikzlibrary{trees}
\usetikzlibrary{shapes}
\usetikzlibrary{chains}
\usetikzlibrary{patterns}

\usepackage{calc}
\usepackage[textwidth=2cm,textsize=small]{todonotes}
\usepackage{algorithm}
\usepackage{algpseudocode}
\usepackage{paralist}
\usepackage{enumitem}
\usepackage{balance}
\usepackage{footnote}
\usepackage{cite} 
\makesavenoteenv{tabular}

\setitemize{noitemsep, topsep=6pt, parsep=1pt, partopsep=6pt, leftmargin=12pt, itemindent=0pt}

\usepackage{xpatch}
\usepackage{textcase}
\makeatletter
\xpatchcmd{\@sect}{\uppercase}{\MakeTextUppercase}{}{}
\xpatchcmd{\@sect}{\uppercase}{\MakeTextUppercase}{}{}
\makeatother

\newcommand{\cO}{\mathcal{O}}

\newcommand{\dset}{\mathbf{D}}
\newcommand{\aset}{\mathbf{A}}

\newcommand{\vset}{\mathbf{X}}

\newcommand{\cset}{\mathbf{C}}
\newcommand{\lset}{\mathbf{L}}
\newcommand{\tuples}{\operatorname{tuples}}
\newcommand{\arity}{\operatorname{arity}}
\newcommand{\streams}{\operatorname{streams}}

\newcommand{\att}{\operatorname{att}}

\newcommand{\pset}{\mathcal{P}}
\newcommand{\upset}{\mathcal{U}}
\newcommand{\bpset}{\mathcal{B}}

\DeclareMathOperator{\supp}{supp}
\newcommand{\ext}[1]{\ensuremath{#1^{\operatorname{SO}}}}
\newcommand{\eext}[1]{\ensuremath{#1^{\exists}}}

\newcommand{\children}{\operatorname{children}}
\newcommand{\tree}{\text{tree}}

\DeclareMathOperator{\socel}{\text{SO-CEL}}
\DeclareMathOperator{\focel}{\text{FO-CEL}}

\DeclareMathOperator{\socelfull}{\socel\!+}

\DeclareMathOperator{\ucea}{\text{UCEA}}

\newcommand{\vdef}{\operatorname{vdef}}
\newcommand{\vdefplus}{\operatorname{vdef}_{+}}

\newcommand{\type}{\operatorname{type}}

\newcommand{\TRUE}{\mathtt{TRUE}}
\newcommand{\FALSE}{\mathtt{FALSE}}
\newcommand{\OP}{\mathtt{OP}}

\newcommand{\as}{~\mathtt{AS}~}
\newcommand{\FILTER}{~\mathtt{FILTER}~}
\newcommand{\cor}{~\mathtt{OR}~}
\newcommand{\sq}{\,;\,}
\newcommand{\ks}{+}

\newcommand{\STRICT}{\mathtt{STRICT}}

\newcommand{\cand}{~\mathtt{AND}~}
\newcommand{\call}{~\mathtt{ALL}~}
\newcommand{\cunless}{~\mathtt{UNLESS}~}

\newcommand{\cin}{~\mathtt{IN}~}
\newcommand{\strictsq}{~\mathtt{:}~}
\newcommand{\strictks}{\oplus}
\newcommand{\START}{~\mathtt{START}}

\newcommand{\sFILTER}{\mathtt{FILTER}}
\newcommand{\scor}{\mathtt{OR}}

\newcommand{\scand}{\mathtt{AND}}
\newcommand{\scall}{\mathtt{ALL}}
\newcommand{\scunless}{\mathtt{UNLESS}}

\newcommand{\WHERE}{~\mathtt{WHERE}~}
\newcommand{\AND}{~\mathtt{AND}~}

\newcommand{\scin}{\mathtt{IN}}
\newcommand{\sstrictsq}{:}
\newcommand{\sSTART}{\mathtt{START}}

\newcommand{\cdotlor}{\mathrel{\ooalign{$\lor$\cr\hidewidth\raise.225ex\hbox{$\cdot\mkern3.3mu$}\cr}}}

\newcommand{\cA}{\mathcal{A}}

\newcommand{\cR}{\mathcal{R}}
\newcommand{\cP}{\mathcal{P}}

\newcommand{\bbN}{\mathbb{N}}

\newcommand{\sem}[1]{{\lsem{}{#1}\rsem}}
\newcommand{\ssem}[1]{\sem{#1}^*}
\newcommand{\trans}[2][]{\raisebox{-1pt}[10pt][0pt]{$\overset{#2}{\underset{^{#1}}{\raisebox{0pt}[3pt][0pt]{$\relbar\mspace{-8mu}\longrightarrow$}}}$}}
\newcommand{\amark}{\bullet}
\newcommand{\umark}{\circ}

\newcommand{\run}{\operatorname{Run}}

\newtheorem{proposition}[theorem]{Proposition}

\begin{document}


\author[1]{Alejandro Grez}
\author[1]{Cristian Riveros}
\author[2]{Martín Ugarte}
\author[2]{Stijn Vansummeren}
\affil[1]{Pontificia Universidad Católica de Chile\\
	\texttt{\{ajgrez, cristian.riveros\}@uc.cl}}
\affil[2]{Université Libre de Bruxelles\\
\texttt{\{martin.ugarte, stijn.vansummeren\}@ulb.ac.be}}
\authorrunning{A.\, Grez, C.\, Riveros, M.\, Ugarte, and S.\, Vansummeren}

%



\title{A Second-Order Approach to Complex Event Recognition}
\maketitle


\begin{abstract}
Complex Event Recognition (CER for short) refers to the activity of detecting patterns in streams of continuously arriving data. This field has been traditionally approached from a practical point of view, resulting in heterogeneous implementations with fundamentally different capabilities. The main reason behind this is that defining formal semantics for a CER language is not trivial: they usually combine first-order variables for joining and filtering events with regular operators like sequencing and Kleene closure. Moreover, their semantics usually focus only on the detection of complex events, leaving the concept of output mostly unattended. Overall, this results in a lack of understanding of the expressive power of CER languages, implying also that the way in which operators are defined is sometimes arbitrary.

In this paper, we propose to unify the semantics and output of complex event recognition languages by using second order objects. Specifically, we introduce a CER language called Second Order Complex Event Logic (SO-CEL for short), that uses second order variables for managing and outputting sequences of events. This makes the definition of the main CER operators simple, allowing us to develop the first steps in understanding its expressive power. We start by comparing SO-CEL with a version that uses first-order variables called FO-CEL, showing that they are equivalent in expressive power when restricted to unary predicates but, surprisingly, incomparable in general. Nevertheless, we show that if we restrict to sets of binary predicates, then SO-CEL is strictly more expressive than FO-CEL. Then, we introduce a natural computational model called Unary Complex Event Automata (UCEA) that provides a better understanding of the expressive power, computational capabilities, and limitations of SO-CEL. We show that, under unary predicates, SO-CEL captures the subclass of UCEA that satisfy the so-called \mbox{*-property}. Finally, we identify the operations that SO-CEL is lacking to capture UCEA and introduce a natural extension of the language that captures the complete class of UCEA under unary predicates.
\end{abstract}


\section{Introduction}\label{sec:intro}

Complex Event Recognition (CER) is a key ingredient of many
contemporary Big Data applications that require the processing of
event streams in order to obtain timely insights and implement
reactive and proactive measures. Examples of such applications include
the recognition of: attacks in computer
networks~\cite{DBLP:conf/sigmod/CranorJSS03,DBLP:conf/sigmod/CranorGJSS02};
human activities in video content~\cite{DBLP:conf/cvpr/0003CNRD16}; traffic incidents in smart cities~\cite{DBLP:journals/tkde/ArtikisSP15}; and
opportunities in the stock
market~\cite{DBLP:conf/debs/KolchinskySS15}, among others.

To support the above-mentioned application scenarios, numerous CER
systems and languages have been proposed in the literature---see e.g.,
the surveys \cite{cugola2012, DBLP:conf/debs/ArtikisMUVW17} and the references therein.
As noted by Cugola and Margara~\cite{cugola2012}, however, the literature focuses mostly on the practical system aspects of CER,
resulting in many heterogeneous implementations with sometimes
fundamentally different capabilities. As a result, little work has
been done on the formal foundations of CER. Consequently, and in
contrast to the situation for relational databases, we currently lack
a common understanding of the trade-offs between expressiveness and
complexity in the design of CER languages, as well as an established
theory for optimizing CER queries. In fact, it is rare to even find a formal definition of a CER language, let alone a formal development of its theory.

Towards a better understanding of the formal foundations of CER, a
subset of the authors has recently proposed and studied a formal logic
that captures the core features found in most CER
languages~\cite{GRUpaper}. This logic, which we will call $\focel$ in
this paper, combines the regular expression operators (sequencing, to
require that some pattern occurs before another somewhere in a stream;
iteration, to recognize a pattern a number of times; and disjunction)
with data filtering features as well as limitedD data outputting
capabilities. $\focel$ follows the approach that seems to be taken by
most of the CER literature (e.g.,
\cite{DBLP:conf/edbt/DemersGHRW06,DBLP:conf/cidr/DemersGPRSW07,SASEautomata,SASEcomplexity,Cugola:2012},
see also \cite{cugola2012}) in that data filtering is supported by
binding variables to individual events in the stream, which can later
be inspected by means of one or more predicates. In this respect,
variables in $\focel$ are \emph{first order} variables, since they
bind to individual events. The first-order nature of variables in CER languages found in the literature on which
$\focel$ is based is problematic for two reasons. (1) It interacts
awkwardly with pattern iteration (i.e., Kleene closure): if a variable
is bound inside Kleene closure, what does the variable refer to when
used outside of a Kleene closure? (2) There is an inherent asymmetry
between the objects manipulated by the CER language (i.e. individual
events bound to first order variables) and the objects that are output
by the language (complex events, where one often wants to include the
set of all matched primitive events). Both of these issues cause
$\focel$, and the practical languages on which it is inspired, to have 
rather awkward variable-scoping rules and a sometimes unexpected
semantics. In a sense, the language becomes quite closely tied to a
specific evaluation model. From a query language viewpoint, this is
undesirable since it restricts the declarative nature of the language,
and hence restricts its optimization and evaluation opportunities.

In this paper, we propose to unify the mechanics of data
filtering and output by using \emph{second order variables} that bind sets of events in a stream.  This allows us to introduce a CER language with simple and intuitive semantics, which we call $\socel$. We study the expressive power of this language and its computational capabilities. Our results are as follows.
\begin{itemize}
\item We first compare $\socel$ against $\focel$ and show that they
  are equivalent in expressive power when equipped with the same unary
  predicates but, surprisingly, incomparable when equipped with
  $n$-ary predicates, $n > 1$. In particular, when equipped with sets
  of binary predicates, $\socel$ is strictly more expressive than
  $\focel$. Conversely, when equipped with sets of ternary predicates,
  the languages are incomparable.
  (Section~\ref{sec:comparison}.) 
\item To get a fundamental understanding of the expressive power,
  computational capabilities, and limitations of the basic operators
  of $\socel$ we then restrict our attention to unary predicates. We
  compare $\socel$ with such predicates against a computational model
  for detecting complex events that we call \emph{Unary Complex Event
  Automata} (UCEA for short). We show that, in this setting, $\socel$
  is strictly weaker than UCEA, but captures the subclass of UCEA that
  satisfy the so-called *-property. Intuitively, this property
  indicates that the UCEA can only make decisions based on events that
  are part of the output.  As a by-product of our development we are
  able to show that certain additional CER operators that have been
  proposed in the literature, such as $\mathtt{AND}$ and
  $\mathtt{ALL}$ are non-primitive in $\socel$ while others, such as
  $\mathtt{UNLESS}$, are primitive. (Section~\ref{sec:unary-core-expr})
\item Finally, we identify the operations that $\socel$ is lacking to
  capture UCEA and introduce a natural extension of the language that
  captures the complete class of UCEA under unary predicates. As a
  result we are also able to give insight into the
  $\mathtt{STRICT}$ selection policy that is supported by some
  CER languages. (Section~\ref{sec:capturing-ma}.)
\end{itemize}
We intuitively motivate $\socel$ in Section~\ref{sec:motivation} and give its formal definition in Section~\ref{sec:prelim}.

\medskip
\noindent \textbf{Related Work.} As already mentioned, the focus in the
majority of the CER literature is on the systems aspects of
CER rather than on the foundational aspects. A notable exception is
the work by Zhang et al on Sase$^+$~\cite{SASEcomplexity}, which
considers the descriptive complexity of a core CER language. It is
unfortunate, however, that this paper lacks a formal definition of the
language under study; and ignores in particular the aforementioned
issues related to the scoping of variables under Kleene closure, as
well as the data output capabilities.

While several automata models for CER have been proposed
before~\cite{SASEautomata,SASEcomplexity,cayuga,DBLP:conf/cidr/DemersGPRSW07,Cugola:2012},
these models are all limited in the sense that automata are required
to adhere to strict topological constraints. Often, an automaton needs
to be a single path from initial to final state, possibly with
self-loops on the nodes in the path. In addition, as shown in
\cite{GRUpaper}, there exist simple complex event patterns for which
the corresponding automata in these models are inefficient. In
contrast, our Unary Complex Event Automata do not have topological
constraints, and are inherently efficient to evaluate
(Proposition~\ref{prop:ucea:efficient}).

\newcommand{\backref}[1]{\ensuremath{\&{#1}}}

Extensions of regular expressions with data filtering capabilities
have been considered before outside of the CER context. \emph{Extended
  regular
  expressions}~\cite{DBLP:books/el/leeuwen90/Aho90,DBLP:journals/ijfcs/CampeanuSY03,DBLP:conf/lata/CarleN09,LibkinTV15}, for example, extend the classical regular expressions operating on
strings with variable binding expressions of the form $x\{e\}$
(meaning that when the input is matched, the substring matched by
regular expression $e$ is bound to variable $x$) and variable
backreference expression of the form $\backref{x}$ (referring to the
last binding of variable $x$). Variables binding expressions can occur
inside a Kleene closure, but when referred to, a variable always
refers to the last binding. Extended regular expressions differ from
$\socel$ and $\focel$ in that they operate on finite strings over a
finite alphabet rather than infinite streams over an infinite alphabet
of possible events; and use variables only to filter the input rather
than also using them to construct the output. Regular expressions with
variable bindings have also been considered in the so-called spanners
approach to information
extraction~\cite{DBLP:journals/jacm/FaginKRV15}. There, however, variables are only used to construct the output
and cannot be used to inspect the input. In addition, variable binding
inside of Kleene closures is prohibited.

Languages with second-order variables, such as monadic second order
logic (MSO), are standard in logic and
databases~\cite{libkin2013elements}.  However, to the best of our
knowledge we are not aware of any language that combines regular
operators with second-order variables as $\socel$, neither has been
proposed in the context of CER.



\section{Second-order Variables to the Rescue}\label{sec:motivation}

We dedicate this section to motivate our proposal for using second-order variables in CER, illustrating how this can enrich and simplify the syntax and semantics of a language.

CER languages usually assume that an event is a relational tuple, composed of a type and a set of attributes, and an event stream is an infinite sequence of events. Events are assumed to appear in generation order in the stream. As a running example, suppose that sensors are positioned throughout a farm to detect freezing plantations. Sensors detect temperature and humidity, generating a stream of events of two types, $T$ and $H$, both of which have a \emph{value} attribute that contains the measured temperature or humidity, respectively. An example stream of events indexed by event generation order is depicted in Figure~\ref{fig:stream}.

\begin{figure}
	\centering{
		{\small
			\begin{tabular}{|c|c|c|c|c|c|c|c|c|c|c|c|}\hline
				type  &$T$&$H$&$H$&$T$&$H$&$T$&$H$&$H$&$T$&$H$ & \ldots \\ \hline
				value & -2 & 30& 20& -1& 27& 2 & 45& 50& -2& 65 & \ldots\\ \hline
				index & 0 & 1 & 2 & 3 & 4 & 5 & 6 & 7 & 8& 9 & \ldots \\ \hline
			\end{tabular}
		}
		\caption{A stream $S$ of events measuring temperature ($T$) in Celsius degrees and humidity ($H$) as a percentage of water in the air.}\label{fig:stream}
		\vspace{-.4cm}
	}
\end{figure}

Now, what is a complex event? Most CER frameworks consider, although in general implicitly, that a complex event is a finite sequence of events that represents a relevant pattern for the user. As an example, suppose that plantations might freeze if \emph{after having a temperature below 0 degrees, there is a period where humidity increases until humidity is over $60\%$}. If we represent complex events as mappings from types to the position of relevant events of that type, in Figure~\ref{fig:stream} we would obtain for example that $[T \rightarrow \{3\},\ H \rightarrow \{4, 6, 7, 9\}]$ is a complex event matching our description of a possibly freezing plantation. Naturally, a user would like to be notified as soon as this pattern is detected (i.e. once event $9$ arrives) and receive the corresponding events to analyze them and possibly take actions. It is not hard to see that several complex events could fire the same pattern at the same time. For example, the above situation is also \emph{matched} by $[T \rightarrow \{0\},\ H \rightarrow \{2,4, 6, 7, 9\}]$, which could also be relevant for the user. This illustrates that the output of a CER pattern should be defined by a set of mappings whose image range over sets (i.e. a set of complex events).

Assigning sets of events to types can be, nonetheless, a shallow representation. Naturally, a user might want to make a distinction between events of the same type. For example the complex event $[T \rightarrow \{3\}, H \rightarrow \{4, 6, 7, 9\}]$ does not have the explicit information of which subset of $H$ corresponds to the increasing sequence of humidity measure, and which correspond to the measure above $60\%$ (although in this case this can be deduced from the pattern). A richer representation for the user could be an assignment of the form $[T \rightarrow \{3\},\ H\!S \rightarrow \{4, 6, 7\},\ H\!H \rightarrow \{9\}]$ where $H\!S$ and $H\!H$ are newly defined labels for encoding the \emph{Humidity Sequence} and \emph{High Humidity} sets of events, respectively.

The previous discussion suggests that second-order assignments should be first citizens in CER languages. Furthermore, the semantics of a CER language can be simplified if second-order variables are used for managing these objects. To motivate this, suppose that we want to declare the plantation freezing pattern above in a CER language. A first attempt could be to use a formula like $\varphi_0: = (T; H+; H)\FILTER \sigma(T, H)$, where $\sigma$ enforces that the values of the events satisfy the corresponding conditions. Intuitively, $\varphi_0$ states that we want to see a temperature ($T$) followed by one or more humidities ($H+$), and ending with a humidity ($H$), such that the $\sigma$ condition is satisfied (as it is standard in CER, here the operators ``;'' and ``+'' allow to skip over intermediate events~\cite{cugola2012}). Note that the pattern $\varphi_0$, however, is not making any explicit distinction between the captured humidities. Then, how could $\sigma$ indicate what are the conditions over the humidity sequence and the final humidity? This could easily be achieved if we had richer complex events, like the one mentioned $H\!S$ and $H\!H$. To this end, we include in our language the operator $\scin$. This allows us to rewrite $\varphi_0$ as $\varphi_1 = T; (H+ \!\!\cin H\!S); (H \cin H\!H) \FILTER \sigma(T,H\!S,H\!H)$. Now $\sigma$ has access to the second-order variables $H\!S$ and $H\!H$, making the filtering more natural.



Now, for the plantation freezing pattern, we need second-order predicates to force that (1) all temperatures in $T$ are below $0$ ($T.value < 0$), (2) the set of humidities $H\!S$ is an increasing sequence ($H\!S.incr$), and (3) the humidity in $H\!H$ is above $60\%$ ($H\!H.value \geq 60$). By combining these predicates, we can write our pattern as:
\[
\varphi \ = \  T; (H+ \!\!\cin H\!S); (H \cin H\!H) \FILTER (T.value < 0 \wedge H\!S.incr \wedge H\!H.value \geq 60)
\]
As an example, the complex event $[T \rightarrow \{3\}, H \rightarrow \{4,6, 7, 9\}, H\!S \rightarrow \{4,6, 7\}, H\!H \rightarrow \{9\}]$ will match $\varphi$ when evaluated over $S$ (Figure~\ref{fig:stream}). It is important to remark that predicates here are evaluated over sets of events. For example the predicate $T.value < 0$ will be satisfied if all temperatures in $T$ have value below $0$ (this is called an universal-extension in Section~\ref{sec:comparison}).

Using second-order predicates might seem loose at first, as predicates could specify arbitrary properties. However, the goal of this approach is to separate what is inherent to a CER framework and what is particular to an application. To illustrate this, consider the framework SASE introduced in \cite{SASE, SASEautomata}. 
The plantation freezing pattern in the SASE language can be written as follows (its meaning can be easily inferred from its syntax):
\[
\operatorname{SEQ}(T \; t, H \; h1[], H \; h2) \WHERE t.value < 0 \AND h1[i-1].value < h1[i].value \AND h2.value \geq 60
\]
What is $i$? Is $h1$ a first or second-order variable? This pattern illustrates that built-in filtering capabilities of CER languages often result in ad-hoc syntax and underspecified semantics. Also, the syntax is never general. What happens if instead of an increasing sequence we want to express that the variance of the humidities is less than five? We do not expect the language to be capable of expressing this, but it could be an application-specific requirement. For this reason, we parametrize our language $\socel$ by an arbitrary set of predicates. We believe that filtering should not be a built-in capability but, in contrast, a practical framework should allow to program application-specific sets of filters as an extension.

Having illustrated that complex events are naturally second-order objects, in the next section we formally present the language $\socel$ and discuss how the introduction second-order variables simplifies the syntax and semantics of CER languages.



\section{Second-order Complex Event Logic}\label{sec:prelim}


In this section, we formally define $\socel$, a core complex event
recognition language based on second-order variables. This language is
heavily based on \cite{GRUpaper}, but it presents a simpler formal definition.  We compare against the language of \cite{GRUpaper}
in Section~\ref{sec:comparison}.

\medskip
\noindent {\bf Schemas, Tuples and Streams.}
	Let $\aset$ be a set of \emph{attribute names} and $\dset$ be a set of values. A database schema $\cR$ is a finite set of relation names, where each relation name $R \in \cR$ is associated to a tuple of attributes denoted by $\att(R)$. If $R$ is a relation name, then an $R$-tuple is a function $t:\att(R) \rightarrow \dset$. We say that the type of an $R$-tuple $t$ is $R$, and denote this by $\type(t)=R$. For any relation name~$R$, $\tuples(R)$ denotes the set of all possible $R$-tuples, i.e., $\tuples(R)=\{t:\att(R) \rightarrow \dset\}$.

	Similarly, for any database schema $\cR$, $\tuples(\cR)=\bigcup_{R \in \cR}\tuples(R)$. Given a schema $\cR$, an $\cR$\textit{-stream} $S$ is an infinite sequence $S = t_0 t_1 \ldots$ where $t_i \in \tuples(\cR)$. When $\cR$ is clear from the context, we refer to $S$ simply as a stream. Given a stream $S = t_0 t_1 \ldots$ and a position $i \in \bbN$, the $i$-th element of $S$ is denoted by $S[i]=t_i$, and the sub-stream $t_{i}t_{i+1} \ldots$ of $S$ is denoted by $S_i$. 
	Note that we consider in this paper that the time of each event is given by its index, and defer a more elaborated time model (like~\cite{WhiteRGD07}) for future work.

\medskip
\noindent {\bf $\socel$ syntax.}
We now give the syntax of $\socel$.
Let $\lset$ be a finite set of monadic second-order (SO) variables containing all relation names (i.e. $\cR \subseteq \lset$). 
An SO predicate of arity $n$ is an $n$-ary relation $P$ over sets of
tuples, $P \subseteq (2^{\tuples(\cR)})^n$. We write
$\arity(P)$ for the arity of $P$. Let $\pset$ be a set of SO
predicates. An atom over $\pset$ is an expression of the form
$P(A_1,\dots, A_n)$ where $P\in\pset$ is a predicate of arity $n$,
and $A_1, \dots, A_n\in\lset$ (we also write $P(\bar{A})$ for $P(A_1, \ldots, A_n)$). 
The set of formulas in
$\socel(\pset)$ is given by the following syntax:
\vspace{-.54cm}

\[\varphi \; := \; R \ \mid \ \varphi \cin A \ \mid \ \varphi[A \rightarrow B] \ \mid \ \varphi \FILTER
\alpha \ \mid \ \varphi \cor \varphi \ \mid \ \varphi \sq \varphi \
\mid \ \varphi \ks \]
\vspace{-.54cm}

\noindent Where $R$ ranges over relation names, $A$ and $B$ over labels in $\lset$ and $\alpha$ over atoms
over $\pset$.

\noindent {\bf $\socel$ semantics.}
In order to define the semantics of core formulas, we first need to
introduce some further notation. A \emph{complex event} $C$ is a
function $C: \lset \rightarrow 2^\bbN$ that assigns a finite set $C(A)$
to every $A \in \lset$.  We say that $C$ is trivial if $C(A) =
\emptyset$ for every $A \in \lset$ in which case we denote $C$ by
$\emptyset$.  Notice that every complex event $C$ can be represented
as a finite set of pairs $(A,j)$ such that $(A,j) \in C$ iff $j \in
C(A)$. We make use of both notations indistinctly. The \emph{support}
$\supp(C)$ of $C$ is the set of positions mentioned in $C$, $\supp(C)
= \bigcup_{A \in \lset} C(A)$. For every non-trivial complex event
$C$, we define the maximum and minimum of $C$ as $\max(C) := \max
\supp(C)$ and $\min(C) := \min \supp(C)$, respectively.  For the
special case $C = \emptyset$, we define $\min(\emptyset) = \infty$ and
$\max(\emptyset) = -\infty$.  Then for every two complex events $C_1$
and $C_2$ with $\max(C_1) < \min(C_2)$, we define their concatenation
as the complex event $C_1 \cdot C_2$ such that $(C_1 \cdot C_2)(A) :=
C_1(A) \cup C_2(A)$.  
For every $A \in \lset$, we define the \emph{extended} complex event $C[A]$ such that $C[A](A) = \supp(C)$ and $C[A](X) = C(X)$ for every $X \neq A$.  
Furthermore, for every $A, B \in \lset$  we define the \emph{renamed} complex event $C[A \rightarrow B]$ such that  $C[A \rightarrow B](A) = \emptyset$, $C[A \rightarrow B](B) = C(A) \cup C(B)$ and $C[A \rightarrow B](X) = C(X)$ for every $X \notin\{A,B\}$.
Finally, given a stream $S$, a complex event $C$ naturally
defines the function $C_S: \lset \rightarrow 2^{\tuples(\cR)}$ where
$C_S(A) := \{S[i] \mid i \in C(A)\}$.

\begin{table}
	\setlength{\jot}{5pt}
	\begin{align*}
	\sem{R}(S, i, j) & \ = \  \{(R,j) \ \mid \ \type(S[j]) = R \}\\
	\sem{\varphi \cin A}(S, i, j) & \ = \ \{C[A] \ \mid \ C \in \sem{\varphi}(S,i,j) \}\\
	\sem{\varphi[A \rightarrow B]}(S, i, j) & \ = \ \{C[A \rightarrow B] \ \mid \ C \in \sem{\varphi}(S,i,j) \} \\
	\sem{\varphi \FILTER P(\bar{A})}(S, i, j) & \ = \ \{C \ \mid \ C
	\in \sem{\varphi}(S,i,j) \text{ and } C_S(\bar{A}) \in P \}\\
	\sem{\varphi_1 \cor \varphi_2}(S, i, j) & \ = \ \sem{\varphi_1}(S,i,j) \, \cup \, \sem{\varphi_1}(S,i,j)\\
	\sem{\varphi_1 \sq \varphi_2}(S, i, j) & \ = \ \{C_1 \cdot C_2 \ \mid \ \exists k. \ C_1 \in \sem{\varphi_1}(S,i,k) \text{ and } C_2 \in \sem{\varphi_2}(S,k+1,j) \}\\
	\sem{\varphi \ks}(S, i, j) & \ = \ \sem{\varphi}(S,i,j) \, \cup \, \sem{\varphi \sq \varphi \ks}(S,i,j)
	\end{align*}
	
	\caption{The semantics of $\socel$.}
	\label{tab-semantics}
	
	\vspace{-15pt}
\end{table}

Now we are ready to define the semantics of $\socel$ formulas.
Given a formula $\varphi$, a stream $S$, and positions
$i \leq j$, in
Table~\ref{tab-semantics} we define recursively the set $\sem{\varphi}(S, i, j)$ of all complex events
of $S$ that satisfy $\varphi$, starting the evaluation at position $i$ and ending at $j$.
Observe that, by definition, if $C \in \sem{\rho}(S,i,j)$ then
$\supp(C)$ is a subset of $\{i,\dots, j\}$ and $j \in \supp(C)$ always.
We say that $C$ belongs to the evaluation of $\varphi$ over $S$ at
position $n \in \bbN$, denoted by $C \in \sem{\varphi}_n(S)$, if $C
\in \sem{\varphi}(S, 0, n)$, namely, we evaluate $\varphi$ over $S$
starting from position~$0$. Intuitively, $C \in \sem{\varphi}_n(S)$
signifies that complex event $C$ was recognized in the stream $S =
t_0t_1\dots$ when
having inspected only the prefix $t_0t_1\dots t_n$.

\begin{example}
  Consider the pattern $\varphi$ introduced in Section~\ref{sec:motivation} to detect possible freezing plantations. We illustrate the evaluation of $\varphi$ over the stream $S$ depicted in Figure~\ref{fig:stream}. First of all, note that, although the conjunction of predicates is not directly supported in $\socel$, this can be easily simulated by a nesting of filter operators. 
  Then, for the sake of simplification, we can analyze $\varphi$ by considering each filter separately.
  For the subformula $\varphi_T = T \FILTER T.value < 0$ we can see that (i) $[T\rightarrow \{3\}] \in \sem{\varphi_T}(S, 0, 3)$. 
  On the other hand, the last event (i.e. $9$) is the only event that satisfies $\varphi_H = (H \cin H\!H) \FILTER H\!H.value \geq 60$ and then (ii) $[H \rightarrow \{9\}, H\!H \rightarrow \{9\}] \in \sem{\varphi_H}(S, 8, 9)$.
  Now, the intermediate formula $\varphi_+ = (H\!+ \! \cin H\!S) \FILTER H\!S.incr$ captures a sequence of one or more $H$-events representing an increasing sequence of humidities.
  Because Kleene closure allows for arbitrary events to occur between iterations, these sequences can be selected from the power set of all $H$-events that produced an increasing sequence like, for example, $[H \rightarrow \{4,6,7\}]$ or $[H \rightarrow \{2,4\}]$.
  In particular, we have that (iii) $[H \rightarrow \{4,6,7\}, H\!S \rightarrow \{4,6,7\}] \in \sem{\varphi_+}(S, 4, 7)$.
  Putting together (i), (ii) and (iii) and noticing that $\varphi = \varphi_T ; \varphi_+ ; \varphi_H$, we have that:
   $$
    T \rightarrow \{3\}, H \rightarrow \{4,6,7,9\}, H\!S \rightarrow \{4,6,7\}, H\!H \rightarrow \{9\} \ \in \ \sem{\varphi}_9(S)
  $$
  is a possible output of evaluating $\varphi$ over $S$.
\end{example}
\begin{example}
	As we saw in the previous example, the $\scin$-operator allows to introduce new names to the output keeping the previous names and positions. 
	However, a user could like to remove or, in other words, rename previous label because there are not relevant for the output. For this, $\socel$ includes the renaming operator $\varphi[A \rightarrow B]$. For example, instead of $\varphi$ we can use the formula $\varphi' = T; (H\!+)[H \rightarrow H\!S]; H[H \rightarrow H\!H] \FILTER (T.value < 0 \wedge H\!S.incr \wedge H\!H.value \geq 60)$ and all complex events that satisfy $\varphi'$ will not include $H$ in the output, i.e. they are replaced by $H\!S$ or $H\!H$.
	In particular: 
	$$
	T \rightarrow \{3\}, H\!S \rightarrow \{4,6,7\}, H\!H \rightarrow \{9\} \ \in \ \sem{\varphi'}_9(S)
	$$
	Notice that renaming operators have been used before in databases and relational algebra~\cite{abiteboul1995foundations}, so it is a natural operator for managing labels in $\socel$. 
\end{example}



\section{The Expressiveness of SO variables versus FO variables}\label{sec:comparison}

$\socel$ is a natural extension of the logic proposed
in~\cite{GRUpaper}, which we will refer to as $\focel$ in this paper.
$\socel$ uses second-order variables whereas $\focel$ uses first order
variables instead. Since, in traditional logics, second-order
languages can encode everything a first-order language can, this could
suggest to the reader that $\socel$ is more expressive than $\focel$.
In this section, we show that this is only partially true. We begin
our discussion with a definition of the syntax and semantics of
$\focel$. For a more detailed explanation of $\focel$, as well as
extensive examples, we address the interested reader
to~\cite{GRUpaper}.

\medskip
\noindent {\bf $\focel$ syntax.}
Let $\vset$ be a set of first order variables. 
Given a schema $\cR$, an FO predicate of arity $n$ is an $n$-ary relation $P$ over
$\tuples(\cR)$, $P \subseteq \tuples(\cR)^n$. If  $\pset$ is a set of
FO predicates then  an atom over $\pset$ is an expression $P(x_1, \ldots, x_n)$ with $P \in \pset$ of arity $n$ and $x_1, \ldots, x_n$ FO variables in $\vset$. 
The set of formulas of $\focel(\pset)$ (called CEPL in~\cite{GRUpaper}) over schema $\cR$ is given by the following grammar:
\[\varphi \; := \; R \as x \ \mid \  \varphi \FILTER \alpha \ \mid \ \varphi \cor \varphi \ \mid \ \varphi\sq\varphi \ \mid \ \varphi\ks. \]
Here, $R$ ranges over relation names in $\cR$, $x$ over variables in $\vset$ and $\alpha$ over atoms in $\pset$.

\medskip
\noindent {\bf $\focel$ semantics.}
For the semantics of $\focel$ we first need to introduce the notion of \emph{match}. A match $M$ is defined as a non-empty and finite set of natural numbers. Note that a match plays the same roll as a complex event in $\socel$  and can be considered as a restricted version where only the support of the output is considered. 
We denote by $\min(M)$ and $\max(M)$ the minimum and maximum element of $M$, respectively. 
Given two matches $M_1$ and $M_2$, we write $M_1 \cdot M_2$ for the \emph{concatenation} of two matches, that is, $M_1 \cdot M_2 := M_1 \cup M_2$ whenever $\max(M_1) < \min(M_2)$ and empty otherwise.
Given an $\focel$ formula $\varphi$, we denote by $\vdef(\varphi)$ all variables defined in $\varphi$ by a clause of the form $R \as x$ and by $\vdefplus(\varphi)$ all variables in $\vdef(\varphi)$ that are defined outside the scope of all $\ks$-operators. For example, in the formula $\varphi = (T \as x \sq (H \as y)\ks) \FILTER z.id = 1$ we have that $\varphi$ uses variables $x, y, z$, $\vdef(\varphi) = \{x,y\}$, and $\vdefplus(\varphi) = \{x\}$. 
A valuation is a function $\nu:\vset \rightarrow \bbN$. Given a finite subset $U \subseteq \vset$ and two valuations $\nu_1$ and $\nu_2$, we define the valuation $\nu_1[\nu_2 / U]$ by $\nu_1[\nu_2 / U](x) = \nu_2(x)$ whenever $x \in U$ and $\nu_1[\nu_2 / U](x) = \nu_1(x)$ otherwise. 

Now we are ready to define the semantics of $\focel$. Given a $\focel$-formula $\varphi$, we say that a match $M$ belongs to the evaluation of $\varphi$ over a stream $S$ starting at position $i$, ending at position $j$, and under the valuation $\nu$ (denoted by $M \in \sem{\varphi}(S, i, j, \nu) $) if one of the following conditions holds:
\begin{itemize}
	\item $\varphi=R \as x$, $M = \{\nu(x)\}$, $\type(S[\nu(x)]) = R$ and $i \leq \nu(x) = j$.
	\item $\varphi = \rho \FILTER P(x_1, \ldots, x_n)$, $M \in \sem{\rho}(S, i, j, \nu)$ and $(S[\nu(x_1)], \ldots, S[\nu(x_n)]) \in P$.
	\item $\varphi=\rho_1 \cor \rho_2$ and $M \in \sem{\rho_1}(S, i, j, \nu)$ or $M \in \sem{\rho_2}(S, i, j, \nu)$).
	\item $\varphi = \rho_1 \sq \rho_2$ and there exists $k \in \bbN$ and matches $M_1$ and $M_2$ such that $M = M_1 \cdot M_2$, $M_1 \in \sem{\rho_1}(S,i,k, \nu)$ and $M_2 \in \sem{\rho_2}(S,k+1,j, \nu)$.
	\item $\varphi=\rho\ks$ and there exists a valuation $\nu'$ such that either $M \in \sem{\rho}(S, i, j, \nu[\nu' / U])$ or  $M \in \sem{\rho\sq\rho\ks}(S, i, j, \nu[\nu' / U])$, where $U = \vdefplus(\rho)$.
\end{itemize}
We say that $M$ belongs to the evaluation of $\varphi$ over
$S$ at position $n \in \bbN$, denoted by $M \in \sem{\varphi}_n(S)$, if $M \in
\sem{\varphi}(S, 0, n, \nu)$ for some valuation $\nu$.

\begin{example}
  \label{ex:focel1}
  Consider that we want to use $\focel$ to see how temperature changes
  at some location whenever there is an increase of humidity
  from below $30$ to above $60$. Assume, for this example, that the location of an event
  is recorded in its $\texttt{id}$ attribute and its humidity in its
  $\texttt{hum}$ attribute. Then, using a self-explanatory syntax for
  FO predicates, we would write:
  \begin{equation*}
    [H\as x;(T \as y\FILTER y.\texttt{id}=x.\texttt{id})+;H \as z]
    \FILTER(x.\texttt{hum} < 30\ \land z.\texttt{hum} > 60 \land x.\texttt{id}=z.\texttt{id})
  \end{equation*}
  Inside the Kleene closure, $y$ is always bound to the current event
  being inspected. The filter $y.\texttt{id}=x.\texttt{id}$ ensures
  that the inspected temperature events of type $T$ are of the same
  location as the first humidity event $x$. Note that, in this case,
  the output is a match (set of positions), and includes in particular
  the positions of the inspected $T$ events. 
\end{example}

In order to make a fair comparison between $\focel$ and $\socel$ we
first need to agree how we relate the FO predicates that can
be used in $\focel$ to the SO predicates that can be used in
$\socel$. Indeed, the expressive power of both languages inherently
depends on what predicates they can use, and we need to put them on
equal footing in this respect. In particular, without any restrictions
on the predicates of $\socel$ we can easily write formulas that are
beyond the scope of $\focel$.  For this reason, we will
restrict ourselves to SO predicates coming from the \emph{universal extension}
of FO predicates.  Here, given a FO predicate $P(x_1, \ldots, x_n)$,
we define its \emph{SO-extension} $\ext{P}$ to be the SO predicate of
the same arity as $P$ such that $(S_1,\dots, S_n) \in \ext{P}$ iff
$\forall x_1 \in S_1, \dots, x_n \in S_n. \; (x_1,\dots,x_n) \in
P$. We extend this definition to sets of predicates: if $\pset$ is a
set of FO predicates, $\ext{\pset}$ is the set $\{\ext{P}\mid
P\in\pset\}$.  In what follows we will compare $\focel(\pset)$ to
$\socel(\ext{\pset})$.

\begin{example}\label{ex:so-extention}
  Using the SO-extensions of the unary FO predicates
  (e.g. $X.\texttt{hum} < 30 := \forall x\in X. \ x.\texttt{hum} < 30$) and the binary
  id-comparison predicate (e.g. $X.\texttt{id} = Y.\texttt{id} := \forall x \in X. \forall y \in Y. y.\texttt{id}=x.\texttt{id}$), the $\focel$ expression of
  Example~\ref{ex:focel1} can be written in $\socel$ as:
  $$
    [H\cin X;(T\ks \! \cin Y); H\cin Z]
    \FILTER[X.\texttt{hum} < 30 \land Z.\texttt{hum} >
      60 \land X.\texttt{id} = Y.\texttt{id} \land X.\texttt{id} = Z.\texttt{id}].
  $$
\end{example}

At this point, the reader may wonder why we focus on \emph{universal}
extensions of FO predicates. After all, one could also consider
\emph{existential} extensions of the form $\eext{P}$ where
$(S_1,\dots, S_n) \in \eext{P}$ iff $\exists x_1 \in S_1, \dots, x_n
\in S_n. \; (x_1,\dots,x_n) \in P$. Under this notion, however,
$\socel$ cannot meaningfully filter events captured by a Kleene
closure.
Indeed, if $X.\texttt{id} = Y.\texttt{id}$ is used with an existential semantics in Example~\ref{ex:so-extention}, then it would
include in $Y$ the $T$ events
occurring between the first $H$ event $X$ and the second $H$ event
$Z$, as long as there is one such $T$ event with the same id as the
single event in $X$. 
Therefore, although existential extensions could be useful in some particular CER use-cases, we compare $\focel$ with $\socel$ by considering only universal extensions.


Another difference to be considered is that $\socel$ outputs complex
events (i.e. second-order assignments over positions) and $\focel$
outputs matches (i.e. sets of positions). To meaningfully compare
both, we consider a formula $\varphi \in \socel$ to be equivalent to a
formula $\psi \in \focel$ (denoted by $\varphi \equiv \psi$) iff for
every stream $S$ and every position $n \in \bbN$ it holds that
$\supp(\sem{\varphi}_n(S)) = \sem{\psi}_n(S)$, where
$\supp(\sem{\varphi}_n(S))$ is defined to be
$\supp(\sem{\varphi}_n(S)) = \{\supp(C) \mid C \in
\sem{\varphi}_n(S)\}$. That is, we consider $\varphi$ equivalent to
$\psi$ if we can obtain $\sem{\psi}_n(S)$ by ``forgetting'' the
variables in the complex events of $\sem{\varphi}_n(S)$.

We now compare both languages. We start by showing that if $\upset$ is a set of unary FO predicates, $\focel(\upset)$ and $\socel(\ext{\upset})$ have the same expressive power.

\begin{theorem} \label{theo:unary-fo-so}
	Let $\upset$ be any set of FO unary predicates. For every formula $\psi \in \focel(\upset)$ there exists a formula $\varphi \in \socel(\ext{\upset})$ such that $\psi \equiv \varphi$, and vice versa. 
\end{theorem}

The previous theorem is of particular relevance since it shows that
both languages coincide in a well-behaved core. $\focel$ with unary
predicates was extensively studied in~\cite{GRUpaper} showing
efficient evaluation algorithms and it is part of 
almost all CER languages~\cite{cugola2012}.

Now we show that if we go beyond unary predicates there are $\socel$
formulas that cannot be equivalently defined in $\focel$. Let
$\pset_{=}$ be the smallest set of FO predicates that contains the
equality predicate $x = y$ and is closed under boolean operations.

\begin{theorem} \label{theo:so-notin-fo}
	There is a formula in $\socel(\ext{\pset}_=)$ that cannot be expressed in $\focel(\pset_=)$.
\end{theorem}

An example of a formula that can be defined in $\socel(\ext{\pset}_=)$
but it cannot be defined in $\focel(\pset_=)$ is $\varphi := (R\ks \sq
T\ks) \FILTER R \neq T$, where $X \neq Y$ is defined as
$\forall x\in X.\forall y\in Y. \,(x \neq y)$. Intuitively, an equivalent
formula in $\focel(\pset_=)$ for $\varphi$ would need to compare every
element in $R$ with every element in $T$, which requires a quadratic number of comparisons. One can show that the number of comparison in the evaluation of an $\focel$ formula is at most linear in the size of the output and, thus, $\varphi$ cannot be defined by any formula in $\focel(\pset_=)$. It is important to note that this result shows the limitations of a CEP-language based on FO variables and what can be gained if SO variables are used. 

A natural question at this point is whether $\socel$ can define every
$\focel$ formula. For binary predicates (e.g. $x.\texttt{id} = y.\texttt{id}$)
the answer is positive, as the following result shows.

\begin{theorem} \label{theo:binary-fo-in-so}
	Let $\bpset$ be any set of FO binary predicates closed under complement.
	Then for every formula $\psi \in \focel(\bpset)$ there exists a formula $\varphi \in \socel(\ext{\bpset})$ such that $\psi \equiv \varphi$.
\end{theorem}

It is important to notice that closeness under complement is a mild restriction over $\bpset$. In particular, if the set $\bpset$ is closed under boolean operations (as usually every CEP query language support), the condition trivially holds.

Interestingly, it is not true that $\socel$ is always more expressive than $\focel$. In particular, there exists an $\focel$ formula with ternary predicates that cannot be defined by any $\socel$ formulas. For the next result, consider the smallest set of FO predicates $\pset_{+}$ containing the sum predicate $x = y + z$ that is closed under boolean operations.  

\begin{theorem} \label{theo:fo-notin-so}
	There is a formula in $\focel(\pset_+)$ that cannot be expressed in $\socel(\ext{\pset}_+)$.
\end{theorem}
In the appendix, we show that the formula $R \as x \sq (S \as y \sq T \as z \FILTER (x = y+z))\ks$ in $\focel(\pset_+)$ cannot be defined in $\socel(\ext{\pset}_+)$. This formula \emph{injects} the $x$-variable  inside the Kleene closure in order to check that each pair $(y, z)$ sums $x$. This capability of injecting variables inside Kleene closure cannot be simulated in $\socel$ given that $\socel$ is a composable language.
It is important to recall that this does not occur if binary predicates are used (Theorem~\ref{theo:binary-fo-in-so}), which are of common use in CER.  

%


\section{On the Expressiveness of Unary Formulas}\label{sec:unary-core-expr}

What is the expressiveness of $\socel(\pset)$? Obviously, as already
illustrated in Section~\ref{sec:comparison}, the answer to this
question depends on the predicates that we allow in $\pset$. To get a
first, fundamental understanding of the expressive power of the basic
operators of $\socel$, we will study this question in the
setting where $\pset$ is limited to contain only the simplest kinds of
predicates possible, namely second-order extensions of unary FO
predicates. When $\pset$ is hence limited, we compare
$\socel(\pset)$ against a computational model for detecting complex
events that we call Unary Complex Event Automata ($\ucea$ for short),
defined next.

\subsection{Unary Complex Event Automata}\label{subsec:ma}



Let $\cR$ be a schema and $\upset$ be a set of unary FO predicates over $\cR$.
We denote by $\upset^+$ the closure of $\upset \cup \{\tuples(R) \mid R \in \cR\}$ under conjunction (i.e. intersection). 
A \emph{unary complex event automaton} ($\ucea$) over $\cR$ and $\upset$ is a tuple $\cA = (Q, \Delta, I, F)$ where $Q$ is a finite set of states, $\Delta \subseteq Q \times \upset^+ \times 2^\lset \times Q$ is a finite transition relation, and $I, F \subseteq Q$ are the set of initial and final states, respectively.
Given an $\cR$-stream $S = t_0 t_1 \ldots$, a run $\rho$ of length $n$
of $\cA$ over $S$ is a sequence of transitions
$\rho: q_0 \ \trans{P_0 / L_0} \ q_1 \  \trans{P_1 / L_1} \ \cdots \ \trans{P_n / L_n} \ q_{n+1}$
such that $q_0 \in I$, $t_i \in P_i$ and $(q_i, P_{i}, L_{i}, q_{i+1}) \in \Delta$ for every $i \leq n$.
We say that $\rho$ is \emph{accepting} if $q_{n+1} \in F$, and denote by $\run_n(\cA, S)$ the set of accepting runs of $\cA$ over $S$ of length $n$.
Further, we define the complex event $C_\rho:\lset \rightarrow 2^\bbN$  induced by $\rho$ as $C_\rho(A) = \{i \in [0,n] \mid A \in L_i\}$ for all $A \in \lset$.
Given a stream $S$ and $n \in \bbN$, we define the set of complex events of $\cA$ over $S$ at position $n$ as $\sem{\cA}_n(S) = \{C_\rho \mid \rho \in \run_n(\cA, S) \}$.

$\ucea$ are a generalization of the
\emph{match automata} (MA) introduced in~\cite{GRUpaper}. The main difference is that match automata output \emph{matches}, which
are sets of positions, while $\ucea$ output complex events
(as defined in Section~\ref{sec:prelim}). In particular, $\ucea$ mark events using SO variables in $\lset$, while match automata mark events by using the symbols $\amark$ (add the event to the match) or $\umark$
(do not add the event to the match), respectively. The empty set
$\emptyset$ in $\ucea$ is the analogous of $\umark$-symbol of match
automata meaning that no SO variable is assigned to the position.

$\ucea$ further generalize MA by lifting \emph{structural}
restrictions on the latter. MA required, for example, that every
transition to a final state \emph{mark} the event with $\amark$. No
such restriction exists for $\ucea$. This relaxation increases the
expressibility of the computational model, at the cost of loosing some
closure properties.

\smallskip

\noindent \textbf{Evaluation of of $\ucea$.} Since the goal of CEP in practice is to process events in
high-throughput environments, one would expect computational models
for CEP to be \emph{efficient}. The following proposition shows that
UCEA are inherently efficient under data complexity. A similar
proposition was established in \cite{GRUpaper} for match automata.

\begin{proposition}
  \label{prop:ucea:efficient}
  For every UCEA $\cA$ there exists a RAM algorithm $A$ that maintains
  a data structure $D$ such that: (1) for every $n$, if $A$ has
  processed prefix $t_0,t_1,\dots,t_n$ of stream $S$ then
  $\sem{\cA}_n(S)$ can be enumerated from data structure $D$ with
  constant delay, and (2) it takes $O(1)$ time to update $D$ upon the
  arrival of a new event $t_{n+1}$.
\end{proposition}
Here, constant-delay enumeration of $\sem{\cA}_n(S)$ from $D$ means
that there exists a RAM routine $\texttt{enum}$ that takes $D$ as
input and that enumerates all complex events in $\sem{\cA}_n(S)$
without repetition, such that (1) the time to initialize enumeration,
finalize enumeration, as well as the time spent between finishing the
output of one complex event and starting the output of the next
complex event is constant, and (2) for each output complex event $C$,
the time spent outputting complex event $C$ is linear in the size of
$C$.


\subsection{Expressiveness of Unary  Formulas}\label{subsec:expressiveness}

The following proposition shows that every formula in $\socel$ with
unary extension of FO predicates can be computed by a complex event
automaton with the same set of predicates.

\begin{proposition}
  \label{prop:unarycel-cea}
  Let $\upset$ be a set of unary FO predicates. For every formula
  $\varphi$ in $\socel(\ext{\upset})$ there exists a $\ucea$ $\cA$ over
  $\upset$ such that $\sem{\varphi}_n(S) = \sem{\cA}_n(S)$, for all streams
  $S$ and $n \in \bbN$.
\end{proposition}
The proof is by a straightforward induction on $\varphi$. It is
natural to ask whether the converse of
Proposition~\ref{prop:unarycel-cea} also holds, namely, if every
$\ucea$ $\cA$ over $\upset$ has an equivalent formula in
$\socel(\ext{\upset})$. Here, however, the answer is no, because
$\ucea$ can make decisions based on tuples that are not part of the
output complex event, while formulas cannot.
\begin{figure}
	\begin{center}
		\begin{tikzpicture}[->,>=stealth, semithick, auto, initial text= {}, initial distance= {3mm}, accepting distance= {4mm}, node distance=4cm, semithick]
		\tikzstyle{every state}=[draw=black,text=black,inner sep=0pt, minimum size=8mm]
		\node[initial,state]	(1) 				{$q_1$};
		\node[state]			(2) [right of=1]	{$q_2$};
		\node[accepting,state]	(3) [right of=2]	{$q_3$};
		\path
		(1)
		edge 				node {$\tuples(H) \mid \{H\}$} (2)
		edge [loop above] 	node {$\TRUE \mid \emptyset$} (1)
		(2)
		edge [loop above] node {$\TRUE \mid \emptyset$} (2)
		edge node {$\tuples(T) \mid \emptyset$} (3);
		\end{tikzpicture}
		\vspace{-.1cm}
		\caption{A Unary Complex Event Automaton that has no equivalent formula in $\socel$.}\label{fig:complement_CEA}
		\vspace{-.74cm}
	\end{center}
\end{figure}
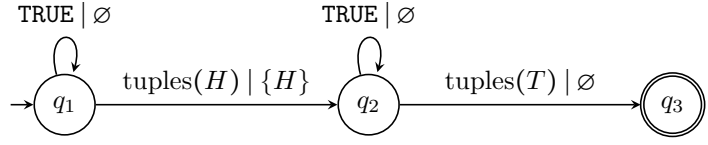
Consider for example the $\ucea$ of Figure~\ref{fig:complement_CEA}. This automaton will output complex events of the form $C = \{ (H,i) \}$, provided that $S[i]$ is of type $H$ and there is a position $j>i$ such that $S[j]$ is of type $T$. It is straightforward to prove that this cannot be achieved by $\socel$ formulas because either such a formula would not check whether $T$ events occurs later, or it would also include the position $j$ in $\supp(C)$ ---which the automaton does not.

In order to capture the exact expressiveness of $\socel$ formulas with
unary predicates, we restrict $\ucea$ to a new semantics called the
\emph{$*$-semantics}.
Formally, let $\cA = (Q, \Delta, I, F)$ be a complex event automaton
and $S=t_1,t_2,\ldots$ be a stream.  A $*$-run $\rho^*$ of $\cA$ over
$S$ ending at position $n$ is a sequence of transitions: $ \rho^*: (q_0,0) \ \trans{P_1/L_1}
\ (q_1,i_1) \ \trans{P_2/L_2} \ \cdots \ \trans{P_m/L_m} (q_m,i_m)$
such that $q_0 \in I$, $0 < i_1 < \ldots < i_m = n$ and, for every $j \geq
1$, $(q_{j-1}, P_j, L_j, q_j) \in \Delta$ with $S[i_j] \in P_j$ and
$L_j \neq \emptyset$.  We say that $\rho^*$ is an accepting $*$-run if
$q_m \in F$. Furthermore, we denote by $C_\rho:\lset \rightarrow
2^\bbN$ the complex event induced by $\rho^*$ as $C_{\rho^*}(A) =
\{i_j \mid A \in L_j\}$ for all $A \in \lset$.  The set of all complex events generated by
$\cA$ over $S$ under the $*$-semantics is defined as: $ \ssem{\cA}_n(S)
= \{C_{\rho^*} \mid \text{$\rho^*$ is an accepting $*$-run of $\cA$
  over $S$  ending at $n$}\} $.  Notice that under this semantics, the automaton no longer has the ability
to verify a tuple without marking it.


We can now effectively capture the expressiveness of unary formulas as follows.
\begin{theorem} \label{theo:CEP_MAs} For every set $\upset$ of unary FO
  predicates, $\socel(\ext{\upset})$ is equally expressive as
  $\ucea(\upset)$ under the $*$-semantics, namely, for every formula
  $\varphi$ in $\socel(\ext{\upset})$, there exists a $\ucea$ $\cA$ over
  $\upset$ such that $\sem{\varphi}_n(S)=\ssem{\cA}_n(S)$ for every 
  $S$ and $n$, and vice versa.
\end{theorem}

For every stream $S$ and complex event $C$, let $S[C]$ refer to the subsequence of $S$ induced by $\supp(C)$.
An interesting property of the $*$-semantics is that, for every $\ucea$ $\cA$, stream $S$, and complex event $C \in \ssem{\cA}(S)$, we can arbitrarily modify, add and remove tuples in $S$ that are not mentioned in $S[C]$, and the original tuples in $S[C]$ would still conform a complex event of $\cA$ over the new stream.
To formalize this, we need some additional definitions.
A \emph{processing-function} $f$ is a function $f:\streams(\cR) \rightarrow 2^\cset$, where $\streams(\cR)$ is the set of all $\cR$-streams and $\cset$ is the set of all complex events (i.e. the set of all finite functions $C: \lset \rightarrow 2^\bbN$).
Although $f$ can be any function that returns set of complex events on input streams, we are interested in the processing-functions $f$ that can be described either by a $\socel$ formula $\varphi$ (i.e. $f = \sem{\varphi}$) or by a $\ucea$ $\cA$ (i.e. $f = \sem{\cA}$).
Let $S_1$, $S_2$ be two streams and $C_1$, $C_2$ be two complex events.
We say that $S_1$ and $C_1$ are \emph{$*$-related} with $S_2$ and
$C_2$, written as $(S_1,C_1)=_*(S_2,C_2)$, if $S_1[C_1] = S_2[C_2]$.

Consider now a match-function $f$. We say that $f$ has the \emph{$*$-property} if, for every stream $S$ and complex event $C \in f(S)$, it holds that $C' \in f(S')$ for every $S'$ and~$C'$ such that $(S,C)=_*(S',C')$.
A way to understand the $*$-property is to see $S'$ as the result of fixing the tuples in $S$ that are part of $S[C]$ and adding or removing tuples arbitrarily, and defining $C'$ to be the complex event that has the same original tuples of~$C$. The following theorem states the relation that exists between the $*$-property and the $*$-semantics over~$\ucea$.

\begin{proposition} \label{prop:MA_MAs}
	If the processing-function defined by a $\ucea$~$\cA$ has the $*$-property, then there exists a $\ucea$ $\cA'$ such that $\sem{\cA}_n(S) = \ssem{\cA'}_n(S)$ for every $S$ and $n$.
\end{proposition}
By combining Theorem~\ref{theo:CEP_MAs} and Proposition~\ref{prop:MA_MAs} we get the following result that shows when a processing function can be defined by unary formula.
\begin{corollary}\label{cor:CEPLs}
  Let $f$ be a processing function. Then $f$ can be defined by a
  $\ucea$ over $\upset$ and has the $*$-property iff there
  exists a formula $\varphi$ in $\socel(\ext{\upset})$ such that $f =
  \sem{\varphi}$.
\end{corollary}

Many query languages for CEP have been proposed in the literature and
most of them include all of the operators of $\socel$ as defined in
Section~\ref{sec:prelim}.  However, some languages include additional
useful operators like $\scand$, $\scall$ and $\scunless$ with the
following semantics. Given a complex event $C$, a stream $S$ and $i, j \in \bbN$:
\begin{itemize}
	\item $C \in \sem{\rho_1 \cand \rho_2}(S, i, j)$ iff $C \in \sem{\rho_1}(S,i,j) \cap \sem{\rho_2}(S,i,j)$.
	\item $C \in \sem{\rho_1 \call \rho_2}(S, i, j)$ iff there exist $i_1,i_2,j_1,j_2 \in \bbN$ and matches $C_1$, $C_2$ such that $C_k \in \sem{\rho_k}(S,i_k,j_k)$, $C = C_1 \cup C_2$, $i = \min\{i_1,i_2\}$ and $j = \max\{j_1,j_2\}$.
	\item $C \in \sem{\rho_1 \cunless \rho_2}(S, i, j)$ iff $C \in
          \sem{\rho_1}(S,i,j)$ and, for every complex event $C'$ and $i',j' \in \bbN$ such that $i \leq i' \leq j' \leq j$, it holds that $C' \notin \sem{\rho_2}(S,i',j')$.
\end{itemize}

Interestingly, from a language design point of view, the operators
$\scand$ and $\scall$ are redundant in the sense that $\scand$ and
$\scall$ do not add expressive power in the unary case. Indeed,
$\scand$ and $\scall$ can be defined by $\ucea$ and both satisfy the
$*$-property. 

\begin{corollary} \label{cor:and_all} Let $\upset$ be a set of unary FO
  predicates. For every expression $\varphi$ of the form
  $\varphi_1~\OP~\varphi_2$, with $\OP \in \{\scand, \scall\}$ and
  $\varphi_i$ in $\socel(\ext{\upset})$, there is a $\socel(\ext{\upset})$
  formula $\varphi'$ such that $\sem{\varphi}_n(S) = \sem{\varphi'}_n(S)$
  for every $S$ and $n$.
\end{corollary}

In contrast, the $\scunless$ operator can be defined by $\ucea$ but one can show that there are formulas mentioning $\scunless$ that do not satisfy the $*$-property. Then, by Corollary~\ref{cor:CEPLs}, $\scunless$ is not expressible in $\socel(\ext{\upset})$ with $\upset$ unary FO predicates. This shows that $\scunless$ adds expressibility to unary $\socel$ formulas while remaining executable by $\ucea$. It makes sense then to include $\scunless$ as a primitive in the design of a CER language. Alternatively, one could envision to extend $\socel$ with operators that make it exactly as expressive as UCEA. We take the latter approach in the next section.


\section{Capturing the Expressive power of UCEA}\label{sec:capturing-ma}

As discussed in Section~\ref{sec:unary-core-expr}, given a set $\upset$
of unary FO predicates, $\socel(\ext{\upset})$ captures the class of
$\ucea$ over $\upset$ that have the $*$-property
(Corollary~\ref{cor:CEPLs}). However,
Proposition~\ref{prop:ucea:efficient} established that all $\ucea$ can
be evaluated efficiently, and not only those satisfying the
$*$-property. It makes sense then to study the origin of this lack of
expressive power and extend the language to precisely capture the
expressiveness of the automata model.

\subsection{Expressibility of UCEA and unary SO-CEL}\label{subsec:expressibility:ucea:socel}
	By taking a close look to the characterization of $\socel$ in terms of the $*$-property, one can easily distinguish three shortcomings of $\socel$ that are not  presented in $\ucea$.
	First, every event that is relevant for capturing a complex event  must be part of the output. Although this might be a desired property in some cases, it certainly prevents a user from describing a formula in which the output is simply a subset of the relevant events. 
	This limitation is explained by the $*$-property, and suggests that to capture $\ucea$ we need an operator that allows to remove or, in other words, project events that must appear in the stream but are irrelevant for the output. 
	Although projection is one of the main operators in relational data management systems, it has not be proposed in the context of CER until now, possibly by the difficulty of defining a consistent semantics that combines projection with operators like Kleene closure. 
	Interestingly, we show below that, by using second-order variables, it is straightforward to introduce a projection operator in $\socel$. 
	
	The second shortcoming of $\socel$ is that it cannot express contiguous sequences of events. Indeed, the sequencing operators ($;$ and $\ks$) allow for arbitrary \emph{irrelevant} events to occur in between. While this is a typical requirement in CER, there  might be cases in which a user wants to capture contiguous events. Indeed, strict sequencing has been included in some CER language before~\cite{SASE} as a so-called \emph{selection operator} that only keeps contiguous sequences of events in the output (see Section~\ref{subsec:strict} for further discussion). Given that this can be naturally achieved by $\ucea$ and has been previously proposed in the literature, it is reasonable to include some operators that allow to declare contiguous sequence of events.

	A final feature that is clearly supported by $\ucea$ but not by $\socel$ is specifying that a complex event starts at the beginning of the stream. This feature is not particularly interesting in CER, but we include it as a new operator with the simple objective of capturing the computational model. Actually, this operator is intensively used in the context of regular expression programing where an expression of the form ``$^\wedge R$'' marks that $R$ must be evaluated starting from the beginning of the document. 
	Therefore, it is not at all unusual to include an operator that recognizes events from the beginning of the stream.

	Given the discussion above, we propose to extend $\socel$ with the following operators:
	\[
	\varphi \;\; := \;\;  \varphi \strictsq \varphi\ \mid \ \varphi \strictks \ \mid \  \pi_L(\varphi)\mid \START(\varphi)
	\]
	where $L \subseteq \lset$. Given a formula $\varphi$ of one of the forms above, a complex event $C$, a stream $S$ and positions $i,j$, we say that $C \in \sem{\varphi}(S, i, j)$ if one of the following conditions holds:
	\begin{itemize}
		\item $\varphi = \rho_1 \strictsq \rho_2$ and $C = C_1 \cdot C_2$ for two nonempty complex events $C_1$ and $C_2$ such that $C_1 \in \sem{\rho_1}(S,i,\max(C_1))$, $C_2 \in \sem{\rho_2}(S, \min(C_2),j)$ and $\max(C_1) = \min(C_2) - 1$.
		\item $\varphi = \rho \strictks$ and either $C \in \sem{\rho}(S,i,j)$ or $C \in \sem{\rho \strictsq \rho \strictks}(S,i,j)$.
		\item $\varphi = \pi_L(\rho)$ and there is $C' \in \sem{\rho}(S,i,j)$ such that $C(A) = C'(A)$ if $A \in L$ and $C(A) = \emptyset$ otherwise.
		\item $\varphi = \START(\rho)$, $C \in \sem{\varphi'}(S,i,j)$, and $\min(C) = i$.
	\end{itemize}
	
	The idea behind $\sstrictsq$ and $\strictks$ is to simulate $\sq$ and $\ks$, respectively, but imposing that \emph{irrelevant} events cannot occur in between. This allows us to recognize, for example, the occurrence of an event of type $R$ immediately after an event of type $T$ ($\varphi = R \strictsq T$), or an unbounded series of consecutive events of type $R$ ($\varphi = R\strictks$). Note, however, that the operator $\strictks$ does not impose that intermediate events are contiguous. For example the formula $(R;S)\strictks$ imposes that the last event $S$ of one iteration occurs right before the first event $R$ of the next iteration, but in one iteration the $R$ event and the $S$ event do not need to occur contiguously.\par
	
	\begin{example}
		Following the schema of our running example, suppose that we want to detect a period of temperatures below $0^\circ$ and humidities below 40\%, followed by a sudden increase of humidity (above 45\%). Naturally, we do not expect to skip \emph{irrelevant} temperatures or humidities, as this would defy the purpose of the pattern. Assuming that we are only interested in retrieving the humidity measurements, this pattern would be written as follows:
		$$\pi_H[((H \cin X) \cor T)\strictks : (H \cin Y) \FILTER (X.\texttt{value} < 40 \land T.\texttt{value} < 0 \land Y.\texttt{value} > 45)].$$
	\end{example}

	To denote the extension of $\socel$ with a set of operators $\cO$ we write $\socel\cup\cO$. For readability, we use the special notation $\socelfull$ to denote $\socel\cup\{\!\strictsq\!,\strictks,\pi, \!\!\START\}$.
	
	Having defined the previous operators, we proceed to show that for every set $\upset$ of unary predicates, $\socelfull(\ext{\upset})$ captures the full expressive power of $\ucea$ over $\upset$. To this end, we say that a formula $\varphi$ in $\socelfull(\ext{\upset})$ is equivalent to a $\ucea$ $\cA$ over $\upset$ (denoted by $\varphi\equiv\cA$) if for every stream $S$ and $n \in \bbN$ it is the case that $\sem{\cA}_n(S)=\sem{\varphi}_n(S)$. 

	\begin{theorem}\label{thm:equiv:socel:ucea}
		Let $\upset$ be a set of unary FO predicates. For every $\ucea$ $\cA$ over $\upset$, there exists a formula $\varphi\in\socelfull(\ext{\upset})$ such that $\varphi\equiv\cA$. Conversely, for every formula $\varphi\in\socelfull(\ext{\upset})$ there exists a $\ucea$ $\cA$ over $\upset$ such that $\varphi\equiv\cA$.
	\end{theorem}

	The relevance of this result lies on the fact that, as shown in Proposition~\ref{prop:ucea:efficient}, every $\ucea$ can be executed efficiently over every stream, and therefore it makes sense to have a language for CER that provides all the capabilities of this computational model.\par

	
	\subsection{Strict Sequencing versus Strict Selection}\label{subsec:strict}
	
	For recognizing events that occur contiguously we introduced the strict-sequencing operators (i.e. $:$ and $\strictks$) that locally check this condition. 
	These operators are the natural extension of $;$ and $\ks$, and they resemble the standard operators of concatenation and Kleene star from regular expressions. 
	However, to the best of our knowledge strict-sequencing has not been proposed before in the context of CER, possibly because adding this feature to a language might complicate the semantics, specially when combined with other non-strict operators.
	To avoid this interaction, the strict-contiguity selection (strict-selection for short) has been previously introduced in~\cite{SASE} by means of a unary predicate that basically forces a complex event $C$ to capture a contiguous set of events.
	We can formalize this operator in our framework as follows. 
	For any formula $\varphi$ in $\socel$ let $\STRICT(\varphi)$ be the syntax for the strict-selection operator previously mentioned. 
	We say that a complex event $C$ induces an interval if $\supp(C)$ is an interval.  Then, given a stream $S$ and two position $i, j \in \bbN$, we define that $C\in\sem{\STRICT(\varphi)}(S,i,j)$ if $C\in\sem{\varphi}(S,i,j)$ and $C$ induces an interval.

	A reasonable question at this point is whether the same expressiveness results could be obtained with the strict-selection operator $\STRICT$. 
	We answer this question giving evidence that our decision of including strict-sequencing operators instead of strict-selection was correct. We show that strict-sequencing and strict-selection coincide if we restrict our comparison to unary predicates. 
	Surprisingly, if we move to binary predicates, strict-selection is strictly less expressive than strict-sequencing.
	
	
	To study the difference in expressiveness, we say that two formulas $\varphi$ and $\psi$ are equivalent, denoted by $\varphi\equiv\psi$, if $\sem{\varphi}_n(S)=\sem{\psi}_n(S)$ for every stream $S$ and position $n$. 
	At a first sight, the strict-sequencing operators and the strict-selection predicates seems equally expressive since both allows to force contiguity between pair of events. 
	At least, this intuition holds whenever we restrict to unary predicates.
	
	\begin{proposition}\label{prop:strict-unary-equivalence}
		Let $\upset$ be a set of unary SO predicates. For every $\varphi$ in $\socel \cup \{\strictsq, \strictks\}(\upset)$, there exists a formula $\psi$ in $\socel \cup \{\STRICT\}(\upset)$ such that $\varphi\equiv\psi$, and vice-versa.
	\end{proposition}
	
	The connection between both operators change if we move to predicates of higher arity. 
	Note, however, that $\STRICT$ can always be simulated by means of the strict sequencing operators $\!\strictsq\!$ and $\strictks$, no matter which set of predicates are used.
	
	\begin{proposition}\label{prop:strict-contained-in-contiguous}
		For any set $\pset$ of SO predicates and for any formula $\varphi \in \socel\cup\{\STRICT\}(\pset)$ there exists a formula $\psi \in \socel\cup\{\!\strictsq\!,\strictks\}(\pset)$ such that $\varphi\equiv\psi$.
	\end{proposition}
	
	To explain our decision of including the operators $:$ and $\strictks$ instead of $\STRICT$, we study the opposite direction. First, it is not hard to see that the operator $:$ can indeed be simulated by means of the operator $\STRICT$.
	Actually, for any formula $\varphi_1 : \varphi_2$ we can isolate the rightmost and leftmost event definition of $\varphi_1$ and $\varphi_2$ respectively, change $:$ by $;$ and surround it by a $\STRICT$ operator.
	Now, if we include the operator $\strictks$, the situation becomes more complex. In particular, for binary predicates, $\STRICT$ is not capable of simulating the $\strictks$-operator.
	
	\begin{theorem}\label{theo:strict-proper-containment}
		For any set $\pset$ of SO predicates and for any formula $\varphi \in \socel\cup\{:\}(\pset)$ there exists a formula $\psi \in \socel\cup\{\STRICT\}(\pset)$ such that $\varphi\equiv\psi$. 
		In contrast, there exists a set $\pset$ containing a single binary SO predicate and a formula $\varphi\in\socel\cup\{\strictks\}(\pset)$ that is not equivalent to any formula in $\socel\cup\{\STRICT\}(\pset)$.
	\end{theorem}
	
	This last Theorem concludes our discussion on the operators for contiguity, and allows us to argue that including the operators $:$ and $\strictks$ is better than including the unary operator $\STRICT$. It is worth noting that the proof of Theorem~\ref{theo:strict-proper-containment} is a non-trivial result that requires a version of the pumping lemma for $\ucea$; the proof can be found in the appendix.

\section{Conclusions}\label{sec:conclusions}

In this paper we have described a novel approach to Complex Event Recognition, where complex events are represented by means of second-order variables. We introduced the language $\socel$, showing how the use of second-order variables simplifies the definition of CER languages. By having a simple definition, we were able provide fundamental results regarding the expressive power and computational capabilities of this language. We discussed the expressibility of different operators and compared $\socel$ to $\focel$ (a language based on first order variables) and to an automata model called $\ucea$ (which, as we show, can be evaluated efficiently over streams). We proved that the expressive power of $\socel$ and $\focel$ are incomparable in general, and that an extended version of $\socel$ with unary filters captures the full expressive power of $\ucea$.\par
This work has settled some fundamental questions regarding CER languages by focusing on restricted fragments, namely those where filters are universal extensions of unary predicates. Nevertheless, there are common CER features that reside outside these fragments, being a prominent example the correlation of distant events. We intend to continue this work in the future and extend our current understanding of $\socel$ to provide a declarative CER framework that can be applied in practice.



\bibliographystyle{abbrv}

\bibliography{references}

%
\onecolumn
\appendix

\section{Proofs of Section~\ref{sec:comparison}}

\subsection{Proof of Theorem~\ref{theo:unary-fo-so}}


To prove the theorem we use the fact that, when dealing with FO unary predicates, one can always rewrite the formulas so that all predicates are applied at the lower level, directly on the assignments.
In $\focel(\upset)$ formulas, this notion is defined on \cite{GRUpaper} as \emph{locally-parametrized normal form}, or LP normal form.
The syntax of formulas in LP normal form is restricted to the following grammar:
\[\varphi \; := \; R \as x \ \mid \  R \as x \FILTER P_1(x) \land \ldots \land P_k(x) \ \mid \ \varphi \cor \varphi \ \mid \ \varphi\sq\varphi \ \mid \ \varphi\ks \]
Where $R$ is a relation, $x$ is a variable and $P_1,\ldots,P_k$ are predicates of $\upset$.
To simplify the presentation of the proof, when writing a conjunction of predicates on the filters, it is short for a series of nested filters, were each one is one of the predicates.
In \cite{GRUpaper}, they give a construction that, for every $\focel(\upset)$ formula, defines an equivalent formula in LP normal form.

For $\socel(\ext{\upset})$ formulas, we show now that one can rewrite them to get a similar structure by pushing down every predicate.
This is a rather predictable property, since every predicate $P \in \ext{\upset}$ is a universally quantified extension of one in $\upset$, thus if a set $A$ satisfies $P$, then every $A' \subseteq A$ also satisfies $P$.
Following this idea we show that, for every $\varphi \FILTER P(A) \in \socel(\ext{\upset})$, if $\varphi$ is not an atomic formula (i.e. $\varphi \neq R$ and $\varphi \neq R \FILTER P_1(R) \land \ldots \land P_k(R)$), then $P(A)$ can be pushed one level deeper in $\varphi$.
We consider the possible cases of $\varphi$:
\begin{itemize}
	\item If $\varphi = \varphi_1 \ \OP \ \varphi_2$, with $\OP \in \{\cor, \sq\}$, then 
	$$\varphi \FILTER P(A) \; \equiv \; \varphi_1 \FILTER P(A) \ \OP \ \varphi_2 \FILTER P(A).$$
	\item If $\varphi = \varphi_1 \ks$, then $\varphi \FILTER P(A) \equiv (\varphi_1 \FILTER P(A)) \ks$.
	\item If $\varphi = \varphi_1 \cin B$:
	\begin{itemize}
		\item if $B \neq A$, then $\varphi \FILTER P(A) \equiv (\varphi_1 \FILTER P(A)) \cin B$,
		\item if $B = A$, then $\varphi \FILTER P(A) \equiv (\varphi_1 \FILTER P(A_1) \land P(A_2) \land \cdots \land P(A_n)) \cin A$, where $A_1, \ldots, A_n$ are the assigned labels in $\varphi_1$.
	\end{itemize}
	\item If $\varphi = \varphi_1 [B_1 \rightarrow B_2]$:
	\begin{itemize}
		\item if $B_1 \neq A$ and $B_2 \neq A$, then $\varphi \FILTER P(A) \equiv (\varphi_1 \FILTER P(A)) [B_1 \rightarrow B_2]$,
		\item if $B_1 \neq A$ and $B_2 = A$, then $\varphi \FILTER P(A) \equiv (\varphi_1 \FILTER P(B_1) \land P(A)) [B_1 \rightarrow B_2]$,
		\item if $B_1 = A$ and $B_2 \neq A$, then $\varphi \FILTER P(A) \equiv \varphi_1 [B_1 \rightarrow B_2]$,
		\item if $B_1 = B_2 = A$, then the renaming can simply be removed.
	\end{itemize}
\end{itemize}
The correctness of this equivalences follows straightforward from the definition of the semantics.
Then, by using this equivalences one can push all the predicates down, and the syntax of the resulting formula is of the form:
\[\varphi \; := \; R \ \mid \  R \FILTER P_1(R) \land \ldots \land P_k(R) \ \mid \ \varphi \cor \varphi \ \mid \ \varphi\sq\varphi \ \mid \ \varphi\ks \]
where $R$ is a relation and $P_1,\ldots,P_k \in \ext{\upset}$.
Notice that we dropped the $\scin$ operator.
This is because all filters are applied on the assignments, therefore the labels do not change anything to the support of the complex events.
We say a $\socel$ formula with unary predicates is in LP normal form if it has this syntax.

Now the equivalence between $\focel(\cP_1)$ and $\socel(\cP_1^e)$ is straightforward.
First, if we have a formula $\varphi \in \focel(\cP_1)$, we rewrite it as a formula $\varphi'$ in LP normal form and then replace every $R \as x \FILTER P(x)$ with $R \FILTER \ext{P}(R)$, where $\ext{P}$ is the SO extension of $P$.
Clearly, $R \as x \FILTER P(x) \equiv R \FILTER P^e(R)$, and by doing induction over the structure of $\varphi'$, one can show that the resulting formula, call it $\psi$, is equivalent to $\varphi'$, thus $\varphi \equiv \psi$.

Similarly for the other direction, if we have a formula $\psi \in \socel(\cP_1^e)$ one can rewrite it as a $\socel$ formula $\psi'$ in LP normal form that has the same support as $\psi$.
Then, we replace every $R \FILTER \ext{P}(R)$ with $R \as x \FILTER P(x)$, where $x$ is a new variable and where $\ext{P}$ is the SO extension of $P$.
By the same argument above, one can show by induction that the resulting $\focel$ formula, call it $\varphi$, is equivalent to $\psi'$, thus $\psi \equiv \varphi$.

\subsection{Proof of Theorem~\ref{theo:so-notin-fo}}


Here we prove that the formula $\varphi = (R\ks \sq
T\ks) \FILTER R \neq T)$ in $\socel(\ext{\pset_=})$ does not have an equivalent formula in $\focel(\pset_=)$.
Intuitively, an equivalent $\focel(\pset_=)$ formula for $\varphi$ will need to compare every element in $R$ with every element in $T$ (i.e. a quadratic number of comparisons).
In the sequel we show that the number of comparisons in the evaluation of an $\focel(\pset_=)$ formula is at most linear in the size of the output, and therefore, $\varphi$ cannot be defined by $\focel(\pset_=)$.

To formalize the notion of the comparisons associated to an output, we extend the semantics of $\focel$ in the following way.
First, we define a comparing set $O$ as a set of tuples, where the first element is a predicate and the followings are positions.
For example, a valid element of $O$ is $(=,1,3)$, which represents that the events at positions $1$ and $3$ were compared with equality, i.e. $S[1] = S[3]$.
Strictly speaking, we should also add the information about the attributes that were being compared, but we leave that out to keep notation simple.
Now, given a formula $\psi$ in $\focel(\pset)$ a match $M$, a comparing set $O$, a stream $S$ and positions $i,j$, we say that $(M,O) \in \sem{\psi}(S,i,j,\nu)$ if:
\begin{itemize}
	\item $\psi=R \as x$, $M = \{\nu(x)\}$, $\type(S[\nu(x)]) = R$, $i \leq \nu(x) = j$ and $O = \emptyset$.
	\item $\psi = \rho \FILTER P(x_1, \ldots, x_n)$, there exists some comparing set $O'$ such that $(M,O') \in \sem{\rho}(S, i,j, \nu)$, $(S[\nu(x_1)], \ldots, S[\nu(x_n)]) \in P$ and $O = O' \cup \{(P,\nu(x_1), \ldots, \nu(x_n))\}$.
	\item $\psi=\rho_1 \cor \rho_2$ and $(M,O) \in \sem{\rho_1}(S, i,j, \nu)$ or $(M,O) \in \sem{\rho_2}(S, i,j, \nu)$).
	\item $\psi = \rho_1 \sq \rho_2$ and there exist matches $M_1$ and $M_2$ and comparing sets $O_1$ and $O_2$ such that $M = M_1 \cdot M_2$, $O = O_1 \cup O_2$, $(M_1,O_1) \in \sem{\rho_1}(S,i,k,\nu)$ and $(M_2,O_2) \in \sem{\rho_2}(S,k+1,\nu)$, with $k = \max(M_1)$.
	\item $\psi=\rho\ks$ and there exists a valuation $\nu'$ such that either $(M,O) \in \sem{\rho}(S, i, j,\nu[\nu' / U])$ or  $(M,O) \in \sem{\rho\sq\rho\ks}(S, i,j, \nu[\nu' / U])$, where $U = \vdefplus(\rho)$.
\end{itemize}
Notice that we only extended the previous semantics of $\focel$ adding this new notion of comparing set.
Therefore, it is not hard to see that $(M,O) \in \sem{\psi}(S,i,j,\nu)$ implies $M \in \sem{\psi}(S,i,j,\nu)$, and that $M \in \sem{\psi}(S,i,j,\nu)$ implies that there is some $O$ such that $(M,O) \in \sem{\psi}(S,i,j,\nu)$.

Now, we show inductively that for every $\psi$, there exist constants $c$ and $d$ such that if $(M,O) \in \sem{\psi}(S,i,j,\nu)$, then $|O| \leq c|M|+d$, i.e. the size of $O$ is linear in the size of $M$.
\begin{itemize}
	\item If $\psi = R \as x$, then $c = d = 0$.
	\item $\psi = \rho \FILTER P(x_1, \ldots, x_n)$, and $c',d'$ are the constants for $\rho$, then $c = c'$ and $d = d'+1$.
	\item $\psi=\rho_1 \cor \rho_2$, $c_1,d_1$ are the constants for $\rho_1$ and $c_2,d_2$ are the constants for $\rho_2$, then $c = \max(c_1,c_2)$ and $d = \max(d_1,d_2)$.
	\item $\psi = \rho_1 \sq \rho_2$, $c_1,d_1$ are the constants for $\rho_1$ and $c_2,d_2$ are the constants for $\rho_2$, then $c = \max(c_1,c_2)$ and $d = d_1+d_2$.
	\item $\psi=\rho\ks$ and $c',d'$ are the constants for $\rho$, then $c = c'+d'$ and $d = 0$.
\end{itemize}
In particular, this means that for every formula $\psi$ in $\focel(\pset_=)$, there is at most a linear number of comparisons between events, which is what we will exploit next.

By contradiction, assume that there is a formula $\psi$ in $\focel(\pset_=)$ that is equivalent to $\varphi$, and let $c,d$ be the constants that bound the size of the comparing set.
Then, for an arbitrary $n$ consider the stream $S_n = R(1) R(2) \ldots R(n) T(n+1) T(n+2) \ldots T(2n)$, and consider the match $M_n = \{1,2,\ldots,2n\}$.
It is clear that $M_n \in \sem{\psi}(S_n)$.
Therefore, there exist a comparing set $O_n$ such that $(M_n,O_n) \in \sem{\psi}(S_n)$.
Now, define $O_n^= \{(i,j) \mid (P_=,i,j) \in O_n\}$ and $O_n^\neq \{(i,j) \mid (P_\neq,i,j) \in O_n\}$, i.e. the sets of pairs compared with equality and inequality, respectively.
Because $|O_n|$ is linear in $|M_n|$, we know that, if $n$ is sufficiently large, there exist positions $k_1$ and $k_2$ with $1 \leq k_1 \leq n < k_2 \leq 2n$ such that $(k_1,k_2) \notin O_n^\neq$.
Moreover, because all events have different value, we know that $|O_n^=| = 0$ (counting out the pairs of the form $(k,k)$).
Now, define a new stream $S_n'$ the same as $S_n$ but replacing the values of $S_n[k_1]$ and $S_n[k_2]$ with some new value, e.g. $2n+1$.
Then, all the comparisons made while evaluating $(M_n,O_n)$ will still hold for $S_n'$, which means that $(M_n,O_n) \in \sem{\psi}(S'_n,i,j,\nu)$ by following the same evaluation for $S_n$.
But at the same time $k_1,k_2 \in M_n$, $\type(S'_n[k_1]) = R$, $\type(S'_n[k_2]) = T$ and $S'_n[k_1] = S'_n[k_2]$, which means that $M_n \notin \sem{\psi}(S'_n,i,j,\nu)$, reaching a contradiction.

We conclude that there cannot exist a formula $\psi$ in $\focel(\pset_=)$ equivalent to $\varphi$.

\subsection{Proof of Theorem~\ref{theo:binary-fo-in-so}}


To prove the theorem we provide a construction that, for any formula $\varphi \in \focel(\bpset)$, gives a formula $\psi \in \socel(\ext{\bpset})$ that is equivalent to $\varphi$.
For notation, we will write $\rho \subseteq \rho'$ to say that $\rho$ is a subformula of $\rho'$.
Moreover, for a formula $\psi = \psi' \FILTER P(\bar{x}) \subseteq \varphi$ and $x \in \bar{x}$, we denote $\varphi^P_x$ to be the subformula such that $\psi' \subseteq \varphi^P_x \subseteq \varphi$, $x \in \vdef(\varphi^P_x)$ and there is no other $\rho \subset \psi'$ that satisfies the above.
That is, $\varphi^P_x$ is the formula closest to the filter such that $x$ is defined in it.

Now, to simplify the proof we make some assumptions on $\varphi$.
We assume that $\varphi$ is safe, as defined in \cite{GRUpaper}, that is, for every subformula of the form $\varphi_1 \sq \varphi_2$ it holds that $\vdefplus(\varphi_1) \cap \vdefplus(\varphi_2) = \emptyset$.
We can make this assumption because in \cite{GRUpaper} they show that every $\focel$ formula can be rewritten into a safe one.
Without loss of generality, we assume that $\varphi$ has no unary predicate, as any predicate $P(x)$ can be easily simulated with a binary one using something like $P(x,x)$.
Now, the construction is the following.

First, consider any subformula of $\varphi$ of the form $\varphi' \FILTER P(x,y)$ such that neither $x$ nor $y$ are in $\vdefplus(\varphi')$.
We will move $P(x,y)$ up to a position at which at least one of its variables is defined.
For this, we use the notion of ``well-formedness'' defined in \cite{GRUpaper}.
There, they say that a formula is \emph{well-formed} if for every subformula of the form $\rho \FILTER P_1(x_1,\ldots,x_k)$ and every $x_i$, there is another subformula $\rho_{x_i}$ such that $x_i \in \vdefplus(\rho_{x_i})$.
In fact, they use a more strict notion called ``bound'', which says that $x$ must be bounded to $\rho_x$ in the sense that it must be always defined (e.g. cannot appear only at one side of an $\scor$).
However, we will not make use of that property.
We use the fact that $\varphi$ is well-formed, which means that there are formulas $\varphi_x$ and $\varphi_y$ such that $x \in \vdefplus(\varphi_x)$, $y \in \vdefplus(\varphi_y)$ and $\varphi' \subseteq \varphi_x \subseteq \varphi_y \subseteq \varphi$ (wlog, we assume that $\varphi_x \subseteq \varphi_y$).
Then, we can move up $P(x,y)$ by rewriting $\varphi_x$ as $\varphi_x^\top \FILTER P(x,y) \cor \varphi_x^\bot$, where $\varphi_x^\top$  and $\varphi_x^\bot$ are $\varphi_x$ replacing $P(x,y)$ with $\TRUE$ and $\FALSE$, respectively.
The idea is that $\varphi_x^\top \FILTER P(x,y)$ considers the cases where the condition $P(x,y)$ is needed and $\varphi_x^\bot$ adds the cases where it is not.
As an intuition of why this is true, one can easily see that $\sem{\varphi_x^\top \FILTER P(x,y)}(S,i,j\nu) \cup \sem{\varphi_x^\bot}(S,i,j,\nu) = \sem{\varphi_x}(S,i,j,\nu)$.
Now, let $\varphi_1$ be the result of doing this to every predicate of $\varphi$.
Then, $\varphi_1$ is such that for every $\varphi' \FILTER P(x,y) \subseteq \varphi_1$ it holds that $x \in \vdefplus(\varphi')$ or $y \in \vdefplus(\varphi')$.

The intuition for the second step is that, for every filter $P(x,y)$, if at any moment of the evaluation we assigned $x$ and $y$ to two events, then they must satisfy $P(x,y)$.
In order to achieve this, we want to rename each assignment with a new variable.
The problem is that, if we do this right away, then it is not clear if we can replace the variables in the filters.
For example, if we have the formula $(R \as x \cor T \as x) \FILTER P(x)$ and we rename both assignments with variables $x_1$ and $x_2$, then it is not clear which variable we need to use in the predicate $P$.
To avoid this issue, we first rewrite $\varphi_1$ into a form called ``disjunctive-normal form'', defined in \cite{GRUpaper}.
A formula $\varphi$ is in \emph{disjunctive-normal form} (or DNF) if $\varphi = (\varphi_1 \cor \ldots \cor \varphi_n)$, where for each $i \in \{1, \ldots, n\}$, it is the case that:
\begin{itemize}
	\item Every $\scor$ in $\varphi_i$ occurs in the scope of a $\ks$-operator.
	\item For every subformula of $\varphi_i$ of the form $(\varphi_i')\ks$, it is the case that $\varphi_i'$ is in DNF.
\end{itemize}
In \cite{GRUpaper} they show that every formula $\varphi$ in $\focel$ can be translated into DNF.
Now that the formula is safe and in DNF, we ensure that for every filter $P(x,y)$ the paths in the parse tree from the filter to the assignment of $x$ and $y$ never get inside an $\scor$.
This does not mean, however, that the assignments are not inside an $\scor$ in the whole formula.
For example, $\rho = (R \as x \FILTER P(x,y) \sq S \as y) \cor T \as y$ is a possible formula at this point, even though $y$ appears in two sides of an $\cor$.
Instead, what we ensure is that for every filter, there is exactly one assignment for each of its variables.
In the case of $\rho$, the filter $P(x,y)$ can reach only the $x$ in $T \as x$ and the $y$ in $S \as y$, but not the one in $T \as y$.
Now that for every filter there is exactly one assignment for each of its variables, we can then rename the variables safely.
For this, identify each assignment $R \as x$ of $\varphi_1'$ with a unique id $i$, and for every filter $P(x,y)$ let $i^P_x$ and $i^P_y$ be the ids of the assignments reachable from that filter.
Then, we rewrite $\varphi_1'$ using a new set of variables $\{x_1,x_2, \ldots \}$ in the following way:
\begin{itemize}
	\item Replace each assignment $R \as y$ with $R \as x_i$, where $i$ is the id of the assignment,
	\item Replace each filter $P(x,y)$ with $P(x_{i^P_x},x_{i^P_y})$.
\end{itemize}
Call $\varphi_2$ the resulting formula.
Because $\varphi_1'$ was safe and in DNF, then the renaming does not change the semantics.

The final step is to turn $\varphi_2$ into a $\socel$ formula.
After turning $\varphi$ into $\varphi_2$, we claim that now we can do it safely by pushing each predicate up until it reaches a point where all its variables are assigned.
Formally, considering labels $\{A_1, A_2, \ldots\}$, what we do is:
\begin{itemize}
	\item Replace each assignment $R \as x_i$ with $R \cin A_i$,
	\item For each subformula with a filter $P(x_i,x_j)$, remove de filter and instead add the filter $\ext{P}(A_i,A_j)$ at formula $\varphi^P_{x_j}$ (assuming $\varphi^P_{x_i} \subseteq \varphi^P_{x_j}$).
\end{itemize}
Then, we define $\psi$ as the resulting formula.
The intuition of why $\psi$ keeps the same semantics of $\varphi_2$ is the following.
At $\varphi_2$ we know that for every $\varphi' \FILTER P(x_i,x_j) \subseteq \varphi_2$, one variable (e.g. $x_i$) is assigned in, and only in $\varphi'$, while the other is assigned somewhere else in the formula, and only there.
Then, for every evaluation that at some point needs to assign $x_i$ to some events $e_i$, it must get inside of $\varphi'$ (because $x_i$ is only there), thus it must first satisfy the filter, i.e. $P(e_i,e_j)$ must hold, where $e_j$ is the current assignment of $x_j$.
Moreover, $x_j$ is only named once in $\varphi^P_{x_j}$, and since $x_j$ is defined in $\varphi^P_{x_j}$ (it cannot be inside a $\ks$), it holds that every assignment of $x_i$ (which is only one) and every assignment of $x_j$ must satisfy $P(x_i,x_j)$.
Since all the assignments of $x_i$ and $x_j$ were labelled with $A_i$ and $A_j$, respectively, this is exactly what $\ext{P}(A_i,A_j)$ in $\psi$ represents.
Thus, we claim that the resulting formula $\psi$ is equivalent to $\varphi$.

\subsection{Proof of Theorem~\ref{theo:fo-notin-so}}

To prove the theorem, we show that the formula: 
$$
\varphi \ = \ R \as x \sq ((S \as y \sq T \as z) \FILTER (x = y+z))\ks$$ 
in $\focel(\pset_+)$ is not expressible in $\socel(\ext{\pset}_+)$.
We begin by giving some definitions we will need next.
We define an \emph{evaluation tree} $T$ as an ordered unranked tree where each node $t$ has a complex event $C_t$ associated to it.
For each node $t$, we refer to the children of $t$ as $\children(t) = (t_1, t_2, \ldots t_k)$.
Notice that we do not bound the number of children of a node, and that there is an order between the children.
We often refer to $T$ by its root node, e.g. if $t$ is the root of $T$, then $C_T = C_t$ and $\children(T) = \children(t)$.
With the notion of evaluation tree, we extend the semantics of $\socel$ in the following way.
We say that $T$ belongs to the evaluation trees of $\varphi$
over $S$ starting at position $i$ and ending at $j$ (denoted by $T \in
\sem{\varphi}^\tree(S, i, j) $) if one of the following conditions holds:
\begin{itemize}
	\item $\varphi = R$, $T$ is a single node and $C_T \in \sem{\varphi}(S,i,j)$.
	\item $\varphi = \rho \cin A$, $\children(T) = (t)$, $C_t \in \sem{\rho}^\tree(S,i,j)$, $C_T(A) = \supp(C_t)$ and $C_T(X) = C_t(X)$ for every $X \neq A$.
	\item $\varphi = \rho[A \rightarrow B]$, $\children(T) = (t)$, $t \in \sem{\rho}^\tree(S,i,j)$, $C_T(A) = \emptyset$, $C_T(B) = C_t(A) \cup C_t(B)$, and $C_T(X) = C_t(X)$ for every $X \notin\{A,B\}$.
	\item $\varphi = \rho \FILTER P(X_1,\dots,X_n)$, $T \in \sem{\rho}(S,i,j)$ and $((C_T)_S(X_1),\dots,(C_T)_S(X_n)) \in P$.
	\item $\varphi = \rho_1 \cor \rho_2$ and $T \in \sem{\rho_1}^\tree(S,i,j) \cup \sem{\rho_1}^\tree(S,i,j)$
	\item $\varphi = \rho_1 \sq \rho_2$, $\children(T) = (t_1,t_2)$ and there exists $k \in \bbN$ such that $C_T = C_{t_1} \cdot C_{t_2}$, $t_1 \in \sem{\rho_1}^\tree(S,i,k)$ and $t_2 \in \sem{\rho_2}^\tree(S,k+1,j)$.
	\item $\varphi = \rho \ks$, $\children(T) = (t_1,t_2,\ldots, t_k)$ and $t_i \in \sem{\rho}^\tree(S,k_{i-1}+1,k_i)$, where $k_0 = i-1$ and $k_i = \max(C_{t_i})$.
\end{itemize}
The idea is the same as for the original semantics, and thus it is easy to see that:
\begin{itemize}
	\item If $T \in \sem{\varphi}^\tree(S,i,j)$ then $C_T \in \sem{\varphi}(S,i,j)$, and
	\item If $C \in \sem{\varphi}(S,i,j)$ then there exists some $T \in \sem{\varphi}^\tree(S,i,j)$ with $C_T = C$.
\end{itemize}
We use the evaluation tree because it gives us more information about how the complex event was evaluated than the complex event itself.

The following lemmas exhibit some interesting properties about the evaluation trees, which can be easily proven:

\begin{lemma}
	For every formula $\varphi$ in $\socel$ there exists some $N$ such that every $T \in \sem{\varphi}^\tree(S,i,j)$ is of depth at most $N$, for any stream $S$ and positions $i,j$.
\end{lemma}

\begin{lemma}
	Consider a formula $\varphi$ in $\socel$, a stream $S$ and positions $i,j$.
	For every $T \in \sem{\varphi}^\tree(S,i,j)$ and $k \in \supp(C_T)$ there is exactly one leaf $t$ in $T$ such that $k \in \supp(C_t)$.
	Moreover, the only nodes $t'$ in $T$ with $k \in \supp(C_{t'})$ are the ones in the path between $t$ and $T$.
\end{lemma}

Now we are ready to prove that $\varphi$ does not have an equivalent formula in $\socel(\ext{\pset}_+)$.
By contradiction, assume that there exists such formula, call it $\psi$.
Let $D$ be the maximum depth of the evaluation trees of $\psi$, and let $K$ be the number of times that the sum predicate $X = Y + Z$ appears in $\psi$.
Now, for an arbitrary $N$ consider the stream:
\[
S = \begin{array}{cccc}
\begin{array}{ccccc}
R & S & T & S & T \\
2N & 1 & 2N-1 & 3 & 2N-3
\end{array} & \cdots & \begin{array}{cc}
S & T \\
N-1 & N+1
\end{array} & \cdots
\end{array}
\]
Intuitively, we chose this stream because the only triples that satisfy $X = Y+Z$ are the ones where the set $X$ is associated to the only one event $R$ and the sets $Y$ and $Z$ are associated, one to only one event $S$, and the other to only the event $T$ that is after that $S$.
Clearly the match $M = \{1,2, \ldots, N,N+1\}$ is in $\sem{\varphi}(S,1,N+1)$.
Therefore, there must be some tree $T \in \sem{\psi}^\tree(S,1,N+1)$ such that $\supp(C_T) = M$.
Let $t$ be the leaf that contains the position $1$ (i.e. the only $R$-tuple), and let $t_1,t_2,\ldots,t_d$ be the nodes in the path from $t$ to $T$ ($d \leq D$).
We know that the sum predicate was applied at most $K$ times in the path between $t$ and $T$.
Moreover, for every occurrence of the sum predicate $A_i = B_i + C_i$ at some node $t_1,\ldots,t_d$, it can only be satisfied if $|C_{t_i}(A_i)| = |C_{t_i}(B_i)| = |C_{t_i}(C_i)| = 1$, and, in particular, if $C_{t_i}(A_i) = \{1\}$, $C_{t_i}(B_i) = \{r_1\}$ and $C_{t_i}(C_i) = \{r_2\}$ for some $r_1,r_2$ at distance 1.
We do not consider the case where some of the labels are mapped to $\emptyset$ because in that case the predicate is not filtering anything.
As stated before, it is easy to see that any other scenario does not satisfy the predicate.

Define $O$ as the set that contains all positions that were compared with $S[1]$, i.e. $O = \{l \mid \exists i. \ l \in B_i \cup C_i\}$.
Since there were at most $K$ occurrences of the sum predicate, we know that $|O|$ is at most $2K$.
Then, if we choose an $N$ big enough (e.g. $N = 2K+2$), we can find two positions $i,j$ that are in $M$ but not in $O$.
Moreover, because of the structure of the stream, we can find $i$ and $j$ such that $j = i+1$.
Intuitively, this means that neither $S[i]$ nor $S[j]$ were compared with $S[1]$.
Furthermore, because no other combination of positions satisfy the sum predicate, we know that $S[i]$ and $S[j]$ were not compared with any other event.
Then, we can define a new stream $S'$ the same as $S$ but replacing the values of $S[i]$ and $S[j]$ with any other value, say $0$, and the evaluation tree $T$ would still be in $\sem{\psi}^\tree(S',1,N+1)$, thus $C_T \in \sem{\psi}(S',1,N+1)$ and $\supp(C_T) = M \in \sem{\varphi}(S',1,N+1)$.
But since the sum of the values $S[i]$ and $S[j]$ is not $N$, we know that $M$ cannot be in $\sem{\varphi}(S',1,N+1)$, reaching a contradiction.


\section{Proofs of Section~\ref{sec:unary-core-expr}}

\subsection{Proof of Proposition~\ref{prop:ucea:efficient}}

In \cite{GRUpaper} there is an algorithm to evaluate match automata, which fulfils the conditions of the proposition.
By making some minor modifications to that algorithm we define an analogous algorithm to evaluate $\ucea$ efficiently.

In Algorithm~\ref{alg:ucea-eval}, procedure \textit{Eval} evaluates a $\ucea$ $\cA$ over a stream $S$, i.e., stores the information of its execution and enumerates the complex events after every new event arrives.
To keep the algorithm simple, we assumed that $\cA$ is I/O deterministic.

\begin{algorithm}\scriptsize
	\caption{Evaluate $\cA = (Q, \delta, q_0, F)$ over a stream $S$}\label{alg:ucea-eval}
	\begin{algorithmic}[1]
		\Procedure{Eval[$\cA$]}{$S$}
		\ForAll{$q \in Q \setminus \{q_0\}$}
		\State $last_q \gets first_q \gets {\tt null}$\hspace*{6em}
		\rlap{\smash{$\left.\begin{array}{@{}c@{}}\\{}\\{}\\{}\end{array}\color{red}\right\}
				\color{red}\begin{tabular}{l}All lists begin empty, except the one of $q_0$, \\ which begins with a $\bot$ and represents the \\ starting point of the complex event.\end{tabular}$}}
		\EndFor
		\State $last_{q_0} \gets first_{q_0} \gets \bot$
		\While{$e \gets {\tt yield}_S$}
		\ForAll{$q \in Q$}
		\State $last'_q \gets first'_q \gets {\tt null}$\hspace*{4.5em}
		\rlap{\smash{$\left.\begin{array}{@{}c@{}}\\{}\\{}\end{array}\color{red}\right\}
				\color{red}\begin{tabular}{l}These variables will represent \\ the lists for the next iteration, \\ and are first initialized empty.\end{tabular}$}}
		\EndFor
		\ForAll{$\{p \mid first_p \neq {\tt null}\}$}
		\ForAll{$\{(L,q) \mid  e \models \alpha \land q = \delta(p,\alpha,L)\}$}
		\If{$L \neq \emptyset$}
		\State ${\tt MoveLabelling}(p,e.{\tt time},L,q)$
		\Else\hspace*{16em}
		\rlap{\smash{$\left.\begin{array}{@{}c@{}}\\{}\\{}\\{}\\{}\\{}\\{}\end{array}\color{red}\right\}
				\color{red}\begin{tabular}{l}This part assembles the lists \\ for the next iteration.\end{tabular}$}}
		\State ${\tt MoveNotLabelling}(p,e.{\tt time},q)$
		\EndIf
		\EndFor
		\EndFor
		\ForAll {$q \in Q$}
		\State $first_q \gets first'_q$\hspace*{6.9em}
		\rlap{\smash{$\left.\begin{array}{@{}c@{}}\\{}\\{}\\{}\end{array}\color{red}\right\}
				\color{red}\begin{tabular}{l}This part updates the lists.\end{tabular}$}}
		\State $last_q \gets last'_q$
		\EndFor
		\State ${\tt Enumerate}(e.{\tt time})$
		\EndWhile
		\EndProcedure
		\Procedure{{\tt MoveLabelling}}{$p,t,L,q$}
		\State $n \gets {\tt Node}(t,L,first_p,last_p,first'_q)$
		\State $first'_q \gets n$\hspace*{12em}
		\rlap{\smash{$\left.\begin{array}{@{}c@{}}\\{}\\{}\\{}\\{}\\{}\\{}\\{}\end{array}\color{red}\right\}
				\color{red}\begin{tabular}{l}This procedure adds the current position \\ and its labels at the top to the list of q.\end{tabular}$}}
		\If{$last'_q = {\tt null}$}
		\State $last'_q \gets first'_q$
		\EndIf
		\EndProcedure
		\Procedure{{\tt MoveNotLabelling}}{$p,t,q$}
		\If{$last'_q = {\tt null}$}
		\State $first'_q \gets first_p$
		\State $last'_q \gets last_p$\hspace*{9.5em}
		\rlap{\smash{$\left.\begin{array}{@{}c@{}}\\{}\\{}\\{}\\{}\\{}\\{}\\{}\\{}\\{}\end{array}\color{red}\right\}
				\color{red}\begin{tabular}{l}This procedure appends the previous list of p \\ at the end of the list of q.\end{tabular}$}}
		\Else
		\State $last'_q.{\tt next} \gets first_p$
		\State $last'_q \gets last_p$
		\EndIf
		\EndProcedure
	\end{algorithmic}
\end{algorithm}
\begin{algorithm}\scriptsize
	\caption{Enumerate complex events at time $time$}\label{alg:ucea-enum}
	\begin{algorithmic}[1]
		\Procedure{{\tt Enumerate}}{$time$}
		\ForAll{$q \in F \cap \{q' \mid first_{q'} \neq {\tt null}\}$}
		\State $n \gets first_q$
		\While{$n \neq {\tt null}$}
		\State ${\tt EnumAll}(n, ())$\hspace*{10.7em}
		\rlap{\smash{$\left.\begin{array}{@{}c@{}}\\{}\\{}\\{}\\{}\\{}\\{}\\{}\\{}\end{array}\color{red}\right\}
				\color{red}\begin{tabular}{l}This part iterates over all the nodes \\ that are on the lists of final states, \\ and applies the {\tt EnumAll} procedure.\end{tabular}$}}
		\State $n \gets n.{\tt next}$
		\EndWhile
		\EndFor
		\EndProcedure
		\Procedure{{\tt EnumAll}}{$node, ce$}
		\If{$node = \bot$}
		\State ${\tt SendResult}(ce)$
		\Else
		\State $ce \gets ce.{\tt add}(node.{\tt time},node.{\tt labels})$
		\State $node' \gets node.{\tt top}$
		\While{$node' \neq node.{\tt bot}$}\hspace*{6em}
		\rlap{\smash{$\left.\begin{array}{@{}c@{}}\\{}\\{}\\{}\\{}\\{}\\{}\\{}\\{}\\{}\\{}\end{array}\color{red}\right\}
				\color{red}\begin{tabular}{l}This part builds all the complex events recursively. \\ Each one is represented by a path beginning in a \\ node in the final state and moving to a previous \\ node until the (initial) node with $\bot$ appears. \\ The only purpose of the last {\tt EnumAll} is to consider \\ the last node on the lists, which is not included \\ in the while.\end{tabular}$}}
		\State ${\tt EnumAll}(node',ce)$
		\State $node' \gets node'.{\tt next}$
		\EndWhile
		\State ${\tt EnumAll}(node',ce)$
		\EndIf
		\EndProcedure
	\end{algorithmic}
\end{algorithm}


In Algorithm~\ref{alg:ucea-eval}, the basic structure to store the complex events' positions is the \textit{node}.
Each node contains five attributes: \textit{time}, \textit{labels}, \textit{top}, \textit{bot}, \textit{next}.
The first one represents the time at which the node was created.
In contrast with the original algorithm, here we had to add the \textit{labels} attribute to store the labels for each position.
Together, the time and labels conform the data used to compute the complex events.
The remaining three are pointers to other nodes, which are better explained later.
The access points to the data structure are the variables $first_q$ and $last_q$.
Consider any iteration $i$.
The main idea of the stored data structure is the following: for every state $q$ there is a list of nodes which contain the complex events' information of all the runs that are in state $q$ while reading the $i$-th event.
The first node of the list is $first_q$, and each node points to the following one with it's attribute $next$.
For notation, let $list_q$ be such list.
For each node $n$ on the list, the attributes $top$ and $bot$ are used to access the previous positions of the complex event, and are used in the following way.
Both of them are in the same previous list, and the former appears first, i.e., $bot$ can be reached by moving through the list from $top$ (using the $next$ attribute).
Let $prev_n = n_1, n_2, \ldots, n_k$ be the nodes between $top$ and $bot$, that is, $n_1 = n.{\tt top}$, $n_k = n.{\tt bot}$ and $n_{i+1} = n_i.{\tt next}$.
Then, each $n_i.{\tt time}$ contains the position that came before $n.{\tt time}$ for some complex event.
That way, all complex events can be computed by recursively moving through all the nodes of $prev_n$.
Finally, considering all the nodes of $list_q$, one can compute all complex events that end at state $q$, which is the main idea of the procedure \textit{EnumAll}.

The updating procedure works in the following way.
The variables $first'_q$ and $last'_q$ are used to store the values of $first_q$ and $last_q$ for the next iteration, i.e., they define the states' lists for the next iteration.
First, if there is a transition $(q,\alpha,L,p)$ such that $e \models \alpha$ and $L \neq \emptyset$, then a new node $n$ is added to the list of $p$.
This node contains the position $e.{\tt time}$, the labels that were assigned to it and pointers to the first and last nodes of the list of $q$ in the previous iteration; these values are stored in the attributes $n.{\tt time}$, $n.{\tt top}$ and $n.{\tt bot}$, respectively.
Also, if there is a transition $(q,\alpha,L,p)$ such that $e \models \alpha$ and $L = \emptyset$, then all the previous list of $q$ is added to the list of $p$.
One way to see this is that all the runs that ended at state $q$ on the previous iteration are now extended with this transition, thus ending now at state $p$.
These two updates are performed by procedures \textit{MoveLabelling} and \textit{MoveNotLabelling}.

The purpose of Algorithm ~\ref{alg:ucea-enum} is to build all the complex events recursively.
Each complex event is represented by a path beginning in a node $n$ in the final state and moving to a previous node $n' \in prev_n$ until the (initial) node with $\bot$ appears.
The only purpose of the last {\tt EnumAll} is to consider the last node on the lists, which is not included in the while.

Consider any $\cA$ and $S$. For every iteration $i$, Algorithm~\ref{alg:ucea-eval} enumerates $\sem{\cA}_i(S)$.
The statement is derived directly from the proof of \cite{GRUpaper}.
Moreover, there the reader can find a proof that the algorithm fulfils conditions  of Proposition~\ref{prop:ucea:efficient}.

\subsection{Proof of Proposition~\ref{prop:unarycel-cea}}

The proposition is deduced directly from the proof of Theorem~\ref{thm:equiv:socel:ucea}.
In the proof, the reader can find a construction that, for any formula $\varphi$ in $\socel(\ext{\cP})$, builds an equivalent $\ucea$ $\cA$ by induction over $\varphi$.

\subsection{Proof of Theorem~\ref{theo:CEP_MAs}}

We begin by considering the first direction, i.e., for every formula $\varphi$ in $\socel(\ext{\upset})$, there exists a $\ucea(\upset)$ $\cA$ such that $\sem{\varphi}_n(S)=\ssem{\cA}_n(S)$ for every stream $S$ and $n \in \bbN$.
To prove this, we use the same construction as the one in Theorem~\ref{thm:equiv:socel:ucea}.
The only difference is for the case of the assignment.
In that case, if $\varphi = R$, then $\cA_\varphi$ is defined as depicted in figure~\ref{fig:R_uceas}, i.e. $\cA_\varphi = (\{q_1,q_2\},\Delta_\varphi,\{q_1\},\{q_2\})$ with $\Delta_\varphi = \{(q_1,\tuples(R),\{R\},q_2)\}$.
For the rest of the cases, i.e. $\cin A$, $[A \rightarrow B]$, $\sq$, $\cor$, $\ks$,$\FILTER$, the construction is the same as in Theorem~\ref{thm:equiv:socel:ucea}.

\begin{figure}
	\begin{center}
		\begin{tikzpicture}[->,>=stealth, semithick, auto, initial text= {}, initial distance= {3mm}, accepting distance= {4mm}, node distance=4cm, semithick]
		\tikzstyle{every state}=[draw=black,text=black,inner sep=0pt, minimum size=8mm]
		\node[initial,state]	(1) 				{$q_1$};
		\node[accepting,state]	(2) [right of=1]	{$q_2$};
		\path
		(1)
		edge 				node {$\tuples(R) \mid \{R\}$} (2);
		\end{tikzpicture}
		\caption{A $\ucea$ for $R$ with the $*$-semantics.}\label{fig:R_uceas}
	\end{center}
\end{figure}
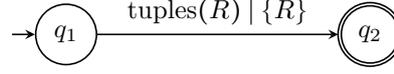

Now, we prove the second direction, that is, for every $\ucea(\upset)$ $\cA$, there exists a formula $\varphi$ in $\socel(\ext{\upset})$ such that $\ssem{\cA}_n(S)=\sem{\varphi}_n(S)$ for every stream $S$ and $n \in \bbN$.
For this, we give a construction that, for any $\cA = (Q, \Delta, I, F)$, defines a formula $\varphi_\cA$.
Consider the set $Q = \{q_1, q_2, \dots, q_n\}$, and assume that $I = \{q_1\}$ and $F = \{q_n\}$.
Here we use the same idea of Theorem~\ref{thm:equiv:socel:ucea}: to define, for every pair of states $(q_i, q_j)$, a formula $\varphi_{ij}$ that represents the complex events of the $*$-runs from $q_i$ to $q_j$.
Furthermore, we define $\varphi_{ij}^k$ the same way but with the restriction that the $*$-runs only pass through states $q_1, \ldots, q_k$.

We define $\varphi_{ij}^k$ recursively in the following way.
The base case $k = 0$ is defined as:
\[
\varphi_{ij}^0 = \left\{
\renewcommand{\arraystretch}{1.3}
\begin{array}{ll}
\rho^{P_1}_{L_1} \cor \rho^{P_2}_{L_2} \cor \ldots \cor \rho^{P_k}_{L_k} & \text{if $(q_i,P_1, L_1,q_j), \ldots , (q_i,P_k, L_k,q_j) \in \Delta$ and $L\neq \emptyset$}\\
\FALSE & \text{otherwise}
\end{array}
\right.
\]
where $\rho^P_L$ is the formula defined in Theorem~\ref{thm:equiv:socel:ucea} that accepts all complex events with a single event that satisfies $P$ and assigns the labels $L$.
Its definition is $\rho^P_L := ( \ldots (\rho_P \cin A_2) \ldots ) \cin A_l$, where $L = \{A_1, \ldots, A_l\}$ and $\rho_P = (R_1 \FILTER P) [R_1 \rightarrow A_1] \cor \cdots \cor (R_r \FILTER P) [R_r \rightarrow A_1]$ for $\cR = \{R_1, \ldots, R_r\}$.
Note that at each $R_i \FILTER P$ we need to remove the predicates of the form $\tuples(R)$, which is done as expected by replacing each one with either $\TRUE$ if $R = R_i$ or $\FALSE$ if $R \neq R_i$.
Next, the recursion is defined as:
\begin{equation*}
\varphi_{ij}^k \; = \; \varphi_{ij}^{k-1} \cor (\varphi_{ik}^{k-1} \sq \varphi_{kj}^{k-1}) \cor (\varphi_{ik}^{k-1} \sq \varphi_{kk}^{k-1}\ks \sq \varphi_{kj}^{k-1})
\end{equation*}
Notice here that the formula differs from the one in Theorem~\ref{thm:equiv:socel:ucea} in that there we use the operator $\sstrictsq$, while here we use the operator $;$.
Finally, the final formula $\varphi_\cA$ is the result of considering $\varphi_{1n}$.
Notice that in contrast with Theorem~\ref{thm:equiv:socel:ucea} here we do not need to apply the $\pi$ operator, since the removal of the type labels $R_i$ was done directly on the assignments.
Moreover, we also do not need to use the $\sSTART$ operator since the $*$-semantics does not allow the automaton to identify the first event of the stream.

\subsection{Proof of Proposition~\ref{prop:MA_MAs}}

Consider any $\ucea(\upset)$ $\cA = (Q, \Delta, I, F)$ that has the $*$-property.
Now, define the $\ucea(\upset)$ $\cA' = (Q, \Delta', I, F)$ such that $\Delta' = \{(p,\alpha,L,q) \mid (p,\alpha,L,q) \in \Delta \land L \neq \emptyset\}$.
Let $S$ be any stream.
We now prove that $\sem{\cA}_n(S) = \ssem{\cA'}_n(S)$.
First, consider a complex event $C \in \sem{\cA}_n(S)$.
This means that there is an accepting run of $\cA$:
$$
\rho: q_0 \ \trans{\alpha_1 / L_1} \ q_1 \  \trans{\alpha_2 / L_2} \ \cdots \ \trans{\alpha_n / L_n} \ q_n
$$
Such that $C_\rho = C$.
Let $\{i_1,i_2, \ldots, i_k\}$ be the support of $C$, and consider the stream $S_C$ as the stream formed by the events $S[i_1], \ldots, S[i_k]$.
Then, because $\cA$ defines a function with $*$-property, there has to be an accepting run of $\cA$ over $S_C$.
$$
\rho': q'_0 \ \trans{\alpha'_1 / L_{i_1}} \ q'_1 \  \trans{\alpha'_2 / L_{i_2}} \ \cdots \ \trans{\alpha'_k / L_{i_k}} \ q'_k
$$
Because of the construction, the analogous $*$-run of $\cA'$ over $S_C$:
$$
\sigma': (q'_0,0) \ \trans{\alpha'_1 / L_{i_1}} \ (q'_1,1) \  \trans{\alpha'_2 / L_{i_2}} \ \cdots \ \trans{\alpha'_k / L_{i_k}} \ (q'_k,k)
$$
is an accepting $*$-run.
Moreover, one can unfold $S_C$ back to the original stream $S$ and the $*$-run:
$$
\sigma: (q'_0,0) \ \trans{\alpha'_1 / L_{i_1}} \ (q'_1,i_1) \  \trans{\alpha'_2 / L_{i_2}} \ \cdots \ \trans{\alpha'_k / L_{i_k}} \ (q'_k,i_k)
$$
is an accepting $*$-run of $\cA'$ over $S$, therefore $C_\sigma = C \in \sem{\cA'}_n(S)$.

The proof for the converse case is practically the same.
Assume that $C \in \ssem{\cA'}_n(S)$, which means that the $*$-run $\sigma$ of $\cA'$ over $S$ exists.
Moreover, the $*$-run $\sigma'$ of $\cA'$ over $S_C$ also exists, thus by the construction of $\cA'$, the run $\rho'$ of $\cA$ must exist.
Because $\cA$ defines a function with $*$-property, the accepting run $\rho$ of $\cA$ over $S$ has to exist.
We conclude that $C_\rho = C \in \sem{\cA}_n(S)$.

\section{Proofs of Section~\ref{sec:capturing-ma}}

\subsection{Proof of Theorem~\ref{thm:equiv:socel:ucea}}

\paragraph{Construction of UCEA}

Let $\upset$ be a set of unary predicates and let $\varphi$ be a formula in $\socelfull(\ext{\upset})$. We start by showing how to construct a $\ucea$ $\cA_\varphi$ over $\upset$ that is equivalent to $\varphi$. We proceed by induction, assuming that for every formula $\psi$ shorter than $\phi$ there is a $\ucea$ $\cA_\psi$ that is equivalent to~$\psi$:

\begin{itemize}
	\item If $\varphi = R$, then $\cA_\varphi$ is defined as depicted in figure~\ref{fig:R_ucea}, i.e. $\cA_\varphi = (\{q_1,q_2\},\Delta_\varphi,\{q_1\},\{q_2\})$ with $\Delta_\varphi = \{(q_1,\TRUE,\emptyset,q_1),(q_1,\tuples(R),\{R\},q_2)\}$
	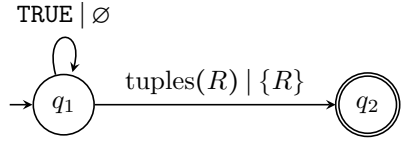
\begin{figure}
		\begin{center}
			\begin{tikzpicture}[->,>=stealth, semithick, auto, initial text= {}, initial distance= {3mm}, accepting distance= {4mm}, node distance=4cm, semithick]
			\tikzstyle{every state}=[draw=black,text=black,inner sep=0pt, minimum size=8mm]
			\node[initial,state]	(1) 				{$q_1$};
			\node[accepting,state]	(2) [right of=1]	{$q_2$};
			\path
			(1)
			edge 				node {$\tuples(R) \mid \{R\}$} (2)
			edge [loop above] 	node {$\TRUE \mid \emptyset$} (1);
			\end{tikzpicture}
			\caption{A $\ucea$ for the atomic formula $R$.}\label{fig:R_ucea}
		\end{center}
	\end{figure}
	\item If $\varphi = \psi \cin A$, then $\cA_\varphi = (Q_\psi, \Delta_\varphi,I_\psi,F_\psi)$ where $\Delta_\varphi$ is the result of adding label $A$ to all non-empty transitions of $\Delta_\psi$.
	Formally, $\Delta_\varphi = \{(p,P,L,q)\in\Delta_\psi \mid L = \emptyset\}\cup\{(p,P,L,q)\mid \exists L'\neq\emptyset \text{ such that } (p,P,L',q)\in \Delta_\psi \land L=L'\cup\{L\}\}$.
	
	\item If $\varphi = \pi_L(\psi)$ for some $L \subseteq \lset$, then $\cA_\varphi = (Q_\psi, \Delta_\varphi,I_\psi,F_\psi)$ where $\Delta_\varphi$ is the result of intersecting each set transition of $\Delta_\psi$ with $L$.
	Formally, that is $\Delta_\varphi = \{(p,P,L \cap L' ,q) \mid (p,P,L',q) \in \Delta_\psi\}$.
	
	\item If $\varphi = \psi[A\rightarrow B]$ for some $A,B \subseteq \lset$, then $\cA_\varphi = (Q_\psi, \Delta_\varphi,I_\psi,F_\psi)$ where $\Delta_\varphi$ is the result of replacing every label $A$ with $B$.
	Formally, if transition $(p,P,L,q)$ is in $\Delta_\varphi$ and $A \in L$, then $(p,P,(L\setminus\{B\})\cup\{A\},q)$ is in $\Delta_\psi$, and if $(p,P,L,q)$ is in $\Delta_\varphi$ and $A \notin L$, then $(p,P,L,q)$ is in $\Delta_\psi$.
	
	\item If $\varphi = \psi \FILTER \ext{P}(A)$ for some unary FO predicate $P$ and $A \in \lset$, then $\cA_\varphi = (Q_\psi, \Delta_\varphi,I_\psi,F_\psi)$ where $\Delta_\varphi$ is defined as
	$\{(p,P',L,q)\in\Delta_\psi \mid A \notin L\}\cup\{(p,P \wedge P',L,q)\mid (p,P',L,q)\in \Delta_\psi \land A \in L\}$. The intuition behind this is that since $\ext{P}$ is the universal extension of $P$, all tuples that are labeld by $A$ must satisfy $P$.
	
	\item If $\varphi = \psi_1 \sq \psi_2$, then $\cA_\varphi = (Q_{\psi_1} \cup Q_{\psi_2}, \Delta_\varphi, I_{\psi_1}, F_{\psi_2})$ where $\Delta_\varphi = \Delta_{\psi_1} \cup \Delta_{\psi_2} \cup \{(p,P,L,q) \mid q \in I_{\psi_2} \land \exists q' \in F_{\psi_1}.(p,P,L,q') \in \Delta_{\psi_1}\}$.
	
	\item If $\varphi = \psi_1 \strictsq \psi_2$, then we do the following.
	We add a new dummy state $q$, which will make the connection between the first part and the last one.
	In order to obtain the $\strictsq$ semantics, we will restrict $q$ to only arrive and depart non-empty transitions.
	We define $\cA_\varphi = (Q_\varphi,\Delta_\varphi,I_\varphi,F_\varphi)$ as follows.
	First, the set of states is $Q_\varphi = Q_{\psi_1} \cup Q_{\psi_2} \cup \{q\}$, where $q$ is a new dummy state.
	Then, the transition relation is $\Delta_{\varphi} = \Delta_{\psi_1} \cup \Delta_{\psi_2} \cup \{(q_1,P,L,q) \mid L \neq \emptyset \land \exists q' \in F_{\psi_1}. \ ((q_1,P,L,q') \in \Delta_{\psi_1})\} \cup \{(q,P,L,q_2) \mid L \neq \emptyset \land \exists q' \in I_{\psi_2}. \ ((q',P,L,q_2) \in \Delta_{\psi_1})\}$.
	Finally, the sets of initial and final states are $I_\varphi = I_{\psi_1}$ and $F_\varphi = F_{\psi_2}$.
	The idea of why this works is that at some point it has to go from $\cA_{\psi_1}$ to $\cA_{\psi_2}$, and the only way to do it is through $q$.
	Then, because $q$ only receives and departs non-empty transitions,  we get the desired result.
	
	\item If $\varphi = \psi \ks$, then $\cA_\varphi = (Q_\psi, \Delta_\varphi, I_\psi, F_\psi)$ where $\Delta_\varphi = \Delta_\psi \cup \{(p,P,L,q) \mid q \in I_\psi \land \exists q' \in F_\psi.(p,P,L,q') \in \Delta_\psi\}$.
	
	\item If $\varphi = \psi \strictks$, then we can use an idea similar to the one we used for the operator $:$.
	We add a new dummy state $q$, which will make the connection between one iteration and the next one.
	In order to obtain the $\strictks$ semantics, we will restrict $q$ to only arrive and depart non-empty transitions.
	We do this as follows: we define $\cA_\varphi = (Q_\psi \cup \{q\}, \Delta_\varphi, I_\psi, F_\psi)$.
	The transition relation is $\Delta_\varphi = \Delta_{\psi_1} \cup \Delta_{\psi_2} \cup \{(q_1,P,L,q) \mid L \neq \emptyset \land \exists q' \in F_\psi. \ ((q_1,P,L,q') \in \Delta_\psi)\} \cup \{(q,P,L,q_2) \mid L \neq \emptyset \land \exists q' \in I_\psi. \ ((q',P,L,q_2) \in \Delta_\psi)\}$.
	The idea of why this works is the same one for the $:$-case: at some point it has to go from one iteration to another, and the only way to do it is through $q$.
	Then, because $q$ only receives and departs non-empty transitions,  we get the desired result.

	\item If $\varphi = \START(\psi)$, then we need to force the first transition to contain at least one label.
	Similar to the previous cases, we add a new dummy state $q$ which will work as our initial state, and we will restrict it to only depart non-empty transitions.
	Formally, $\cA_\varphi = (Q_\psi \cup \{q\}, \Delta_\varphi, \{q\},F_\psi)$, where $\Delta_\varphi = \Delta_psi \cup \{(q,P,L,p) \mid \exists q' \in I_\psi. ((q',P,L,p) \in \Delta_{\psi})\}$.

	\item If $\varphi = \psi_1 \cor \psi_2$, then $\cA_\varphi$ is the automata union between $\cA_{\psi_1}$ and $\cA_{\psi_2}$ as one would expect: $\cA_\varphi = (Q_{\psi_1} \cup Q_{\psi_2}, \Delta_{\psi_1} \cup \Delta_{\psi_2}, I_{\psi_1} \cup I_{\psi_2}, F_{\psi_1} \cup F_{\psi_2})$.
	
\end{itemize}
Notice that the size of the resulting automaton $\cA_\varphi$ is linear in the size of $\varphi$.
Moreover, unlike the construction in \cite{GRUpaper}, where it needed some preprocessing on the formula $\varphi$ (in particular, to push the predicates down), here we do not need any preprocessing because the construction for the case $\sFILTER$ is straightforward.

\subsubsection{From UCEA to unary SO-CEL}

Now we proceed to show the opposite direction. This means that given a UCEA we need to define an equivalent unary SO-CEL formula. Let $\cA = (Q, \Delta, I, F)$ be a UCEA, with $Q = \{q_1, \ldots, q_n\}$. Without loss of generality, assume that there is only one initial state and one final state, i.e., $I = \{q_1\}$ and $F = \{q_n\}$. Moreover, we only consider non-zero executions, that is, an accepting run of a UCEA must be of length at least $1$. Notice that this limits automata to not being able to run over the empty stream, but since we only care about large (potentially infinite) streams, this case is not of interest.

In the sequel we define a formula $\varphi_\cA$ such that $\sem{\cA}(S) = \sem{\varphi_\cA}(S)$ for every $S$. The main idea is based on the construction from finite automata over words which defines, for every pair of states $q_i, q_j$, a $\socel$ formula $\varphi_{ij}$ that represents the complex events defined by the runs from $q_i$ to $q_j$. Furthermore, we define $\varphi_{ij}^k$ the same way but with the restriction that the runs only pass through states $q_1, \ldots, q_k$.
It is clear that $\varphi_{ij}^{|Q|} = \varphi_{ij}$.

To simplify the construction, we give some definitions.
First, we define the $\FALSE$ formula as a formula that is never satisfied.
One way to define it is $\FALSE = (R \FILTER \bot)$, but we will use $\FALSE$ to keep the proof simple.
Also, if $P$ is a formula in $\upset_\plus$ and $R$ is a relation, we define $P[R]$ as the formula that results after replacing in $P$ every occurrence of $\tuples(R)$ with $\TRUE$ and every $\tuples(R')'$ with $\FALSE$ for all $R' \neq R$.
Finally, if $L$ is the set of labels used in $\cA$, we assume that $L \cap \cR = \emptyset$, therefore no label is mistakenly projected.
This is a reasonable assumption, since one can simply add a duplicate $R'$ to every $R \in \cR$, use $R'$ instead of $R$ in $\cA$ to create $\varphi_\cA$, and then rename $R'$ with $R$ over $\varphi_\cA$.

Now we define $\varphi_{ij}^k$ recursively in the following way.
The base case $k = 0$ is defined as:
\[
\varphi_{ij}^0 = \left\{
\renewcommand{\arraystretch}{1.3}
\begin{array}{ll}
\rho^{P_1}_{L_1} \cor \rho^{P_2}_{L_2} \cor \ldots \cor \rho^{P_k}_{L_k} & \text{if $(q_i,P_1, L_1,q_j), \ldots , (q_i,P_k, L_k,q_j) \in \Delta$}\\
\FALSE & \text{otherwise}
\end{array}
\right.
\]
where $\rho^P_L$ represents the CEPL formula that accepts all matches with a single event that satisfies $P$ and assigns the labels $L$.
We define this formula as $\rho^P_L := ( \ldots (\rho_P \cin A_1) \ldots ) \cin A_l$, where $L = \{A_1, \ldots, A_l\}$ and $\rho_P = (R_1 \FILTER P[R_1]) \cor \cdots \cor (R_r \FILTER P[R_r])$ for $\cR = \{R_1, \ldots, R_r\}$.

Next, the recursion is defined as:
\begin{equation}
\varphi_{ij}^k \; = \; \varphi_{ij}^{k-1} \cor (\varphi_{ik}^{k-1} \strictsq \varphi_{kj}^{k-1}) \cor (\varphi_{ik}^{k-1} \strictsq \varphi_{kk}^{k-1}\ks \strictsq \varphi_{kj}^{k-1}) \label{eq:formula}
\end{equation}
Finally, the final formula $\varphi_\cA$ is the result of considering $\varphi_{1n}$ and projecting to retrieve only the labels of $\cA$.
Formally, we define it as $\varphi_\cA := \pi_L(\!\!\!\START( \varphi_{1n} ) ) $, where $L$ is the set of labels in $\cA$.
The $\START$ only forces the match to begin immediately at position $0$, and the projection is needed so that the result does not contain the default labels assigned by the formula.
We assume that $L \cap \cR = \emptyset$, therefore no label is mistakenly projected.

Finally, it is straightforward to prove the correctness of the construction by induction over the number of states.

\subsection{Proof of Proposition~\ref{prop:strict-unary-equivalence}}

Consider a formula $\varphi$ in $\socel\cup\{\strictsq,\strictks\}(\upset)$.
We first prove that there is a formula $\psi$ in $\socel\cup\{\strictsq,\STRICT\}(\upset)$ which is equivalent to $\varphi$, and then the proof follows directly from Theorem~\ref{theo:strict-proper-containment}.

Consider any formula $\varphi'$ in $\socel\cup\{\strictsq,\STRICT\}(\upset)$.
We prove by induction that there exists a formula $\psi'$ in $\socel\cup\{\strictsq,\STRICT\}(\upset)$ which is equivalent to $\varphi' \strictks$.
For the sake of simplicity, we first push all the labellings down to the assignments.
For any formula $\varphi$ and labels $A,B$ we write $\varphi^{A\rightarrow B}$ to refer to the formula $\varphi$ after replacing every occurrence of $A$ by $B$. 
Now, we push the labellings in the following way.
For every subformula of $\varphi'$ with the form $\rho \cin A$, we replace using the following equivalences:
\begin{itemize}
	\item If $\rho = \rho_1 \ \OP \ \rho_2$, with $\OP \in \{\sq, \strictsq, \cor\}$, then $\rho \cin A \equiv (\rho_1 \cin A) \ \OP \ (\rho_2 \cin A)$,
	\item If $\rho = \rho_1 \ks$, then $\rho \cin A \equiv (\rho_1 \cin A) \ks$,
	\item If $\rho = \STRICT(\rho_1)$, then $\rho \cin A \equiv \STRICT(\rho_1 \cin A)$,
	\item If $\rho = \rho_1 [B_1 \rightarrow B_2]$, then $\rho \cin A \equiv \rho_1^{B_1\rightarrow B'} \cin A [B' \rightarrow B_2]$, where $B'$ is a new label,
	\item If $\rho = \rho_1 \FILTER P(B)$, then $\rho \cin A \equiv \rho_1^{B\rightarrow B'} \cin A \FILTER P(B')[B' \rightarrow B]$, where $B'$ is a new label.
\end{itemize}

Now that all the labellings of $\varphi'$ are at the bottom-most level, we proceed to prove by induction that there exists a formula $\psi'$ in $\socel\cup\{\strictsq,\STRICT\}(\upset)$ which is equivalent to $\varphi' \strictks$.
Consider the following cases:
\begin{itemize}
	\item For the base case, if $\varphi' = R$, then $\psi' = \STRICT(R \ks)$ is equivalent to $\varphi' \strictks$.
	\item If $\varphi' = \varphi_1 \cin A$, note that $(\varphi_1 \cin A) \strictks \equiv (\varphi_1 \strictks) \cin A$.
	Then, by induction hypothesis there is a formula $\sigma$ in $\socel\cup\{\strictsq,\STRICT\}(\upset)$ equivalent to $\varphi_1 \strictks$, therefore $\psi' = \sigma \cin A$ is equivalent to $\varphi' \strictks$.
	The case for $\varphi' = \varphi_1 [A \rightarrow B]$ is exactly the same.
	\item If $\varphi' = \varphi_1 \FILTER P(A)$, note that $(\varphi_1 \FILTER P(A)) \strictks \equiv (\varphi_1 \strictks) \FILTER P(A)$.
	Then, by induction hypothesis there is a formula $\sigma$ in $\socel\cup\{\strictsq,\STRICT\}(\upset)$ equivalent to $\varphi_1 \strictks$, therefore $\psi' = \sigma \FILTER P(A)$ is equivalent to $\varphi' \strictks$.
	\item If $\varphi' = \STRICT(\varphi_1)$, then $\psi' = \STRICT(\varphi_1 \ks)$ is equivalent to $\varphi' \strictks$.
	\item If $\varphi' = \varphi_1 \sq \varphi_2$, then $\psi' = \varphi_1 \sq (\varphi_2 \cor ((\varphi_2 \strictsq \varphi_1) \ks \sq \varphi_2))$ is equivalent to $\varphi' \strictks$.
	\item If $\varphi' = \varphi_1 \ks$, then $\psi' = \varphi_1 \ks$ is equivalent to $\varphi' \strictks$.
\end{itemize}
We do not consider the $\strictsq$-case since we know that they can be removed by using $\STRICT$ instead.

The last and more complex operator is the $\scor$, for which we have to consider $\varphi' = \varphi_1 \cor \varphi_2$, with all possible cases for $\varphi_1$ and $\varphi_2$.
The simplest scenario is where both $\varphi_1$ and $\varphi_2$ have either the form $R$, $R \cin A$ or $\STRICT(\psi)$ for some $\psi$, at which case we can simply write $\varphi'\strictks$ as $\STRICT(\varphi' \ks)$ (here we use that all labellings are at the bottom-most level, thus we can avoid that case in the sequel).

Now we consider the cases where some of them does not have this form (w.l.o.g. assume is $\varphi_2$).
Consider the case $\varphi_2 = \rho_1 \sq \rho_2$.
Here we use the following equivalence:

\[
\begin{array}{rcll}
(\varphi_1 \cor (\rho_1 \sq \rho_2)) \strictks &\equiv& \varphi_1 \strictks \cor (\rho_1 \sq \rho_2) \strictks \cor&(1)\\ 
&&(\varphi_1 \strictks \strictsq (\rho_1 \sq \rho_2) \strictks) \strictks \cor & (2)\\
&&\varphi_1 \strictks \strictsq ((\rho_1 \sq \rho_2) \strictks \strictsq \varphi_1 \strictks) \strictks \cor & (3)\\
&&((\rho_1 \sq \rho_2)\strictks \strictsq \varphi_1 \strictks)\strictks \cor& (4)\\
&&(\rho_1 \sq \rho_2)\strictks \strictsq (\varphi_1 \strictks \strictsq (\rho_1 \sq \rho_2))\strictks& (5)
\end{array}
\]
Here, part $(1)$ has no problem since the $\strictks$-operator is applied over subformulas of the original one, thus by induction hypothesis they can be written without $\strictks$.
Moreover, with some basic transformations in part $(2)$ (replacing $(\rho_1 \sq \rho_2)\strictks$ by $\rho_1 \sq (\rho_2 \cor ((\rho_2 \strictsq \rho_1)\ks \sq \rho_2))$) one can show that it is equivalent to the formula $(\sigma_1 \sq \sigma_2) \strictks$, where $\sigma_1 = \varphi_1 \strictks \strictsq \rho_1$ and $\sigma_2 = \rho_2 \cor ((\rho_2 \strictsq \rho_1) \ks \sq \rho_2)$.
Then, we can replace $(\sigma_1 \sq \sigma_2)$ with $\sigma_1 \sq (\sigma_2 \cor ((\sigma_2 \strictsq \sigma_1)\ks \sq \sigma_2))$, and the resulting formula will contain only one $\strictks$ in the form $\varphi_1 \strictks$, which by induction hypothesis can also be removed.
Similarly, parts $(3)$, $(4)$ and $(5)$ can be rewritten this way, therefore for the case of $\varphi_2 = \rho_1 \sq \rho_2$ the induction statement remains true.

Now consider the case $\varphi_2 = \rho\ks$.
Notice that, by expanding the definition of the $\ks$-operator ($\rho \ks \equiv \rho \cor \rho \sq \rho \ks$), $\varphi'$ can then be written as $(\varphi_1 \cor \rho) \cor (\rho \sq \rho \ks)$.
Then, if we redefine $\varphi_1 := (\varphi_1 \cor \rho)$ and $\varphi_2 ;= (\rho \sq \rho \ks)$, clearly $\varphi'$ would have the form $\varphi_1 \cor (\rho_1 \sq \rho_2)$.
Therefore, we can apply the previous case and the resulting formula will still satisfy the induction statement, thus it remains true in the case $\varphi_2 = \rho \ks$.

Consider now the case $\varphi_2 = \rho \FILTER P(A)$.
We use the equivalence 
$$(\varphi_1 \cor \rho \FILTER P(A))\strictks \equiv (\varphi_1 \cor \rho^{A\rightarrow A'})\strictks \FILTER P(A') [A' \rightarrow A]$$
By induction hypothesis we know that there is a formula $\sigma$ equivalent to $(\varphi_1 \cor \rho^{A\rightarrow A'})\strictks$, thus $\varphi'\strictks \equiv \sigma \FILTER P(A') [A' \rightarrow A]$.

The only case that is left is $\varphi_2 = \rho [A \rightarrow B]$, for which we can use the equivalence
$$(\varphi_1 \cor \rho[A \rightarrow B])\strictks \equiv (\varphi_1 \cor \rho^{A\rightarrow A'})\strictks [A' \rightarrow B]$$
By induction hypothesis we know that there is a formula $\sigma$ equivalent to $(\varphi_1 \cor \rho^{A\rightarrow A'})\strictks$, thus $\varphi'\strictks \equiv \sigma [A' \rightarrow B]$.

Then, we can replace every subformula $\varphi'$ of $\varphi$ with its equivalent formula $\psi'$ in a bottom-up fashion to ensure that the subformulas of $\varphi'$ satisfy the condition of being in $\socel\cup\{\strictsq,\STRICT\}(\upset)$.
Finally, the remaining formula $\psi$ does not contain $\strictks$ and thus is in $\socel\cup\{\strictsq,\STRICT\}(\upset)$.

\subsection{Proof of Proposition~\ref{prop:strict-contained-in-contiguous}}

Consider a formula $\varphi$ in $\socel\cup\{\STRICT\}(\pset)$.
We first prove that for every formula $\varphi'$ in $\socel\cup\{\STRICT\}(\pset)$ there is a formula $\psi'$ in $\socel\cup\{\strictsq,\strictks\}(\pset)$ that is equivalent to $\STRICT(\varphi')$, for which we do induction over $\varphi'$.
The base case is $\varphi' = R$, which clearly satisfies the above considering $\psi' = R$.
For the inductive step, consider the following cases.
\begin{itemize}
	\item If $\varphi' = \rho \cin A$, then $\STRICT(\varphi') \equiv \STRICT(\rho) \cin A$.
	By induction hypothesis, there is a formula $\sigma$ in $\socel\cup\{\strictsq,\strictks\}(\pset)$ equivalent to $\STRICT(\rho)$.
	Thus, $\psi' = \sigma \cin A$ is equivalent to $\STRICT(\varphi')$.
	
	\item If $\varphi' = \rho [A \rightarrow B]$, then $\STRICT(\varphi') \equiv \STRICT(\rho) [A \rightarrow B]$.
	By induction hypothesis, there is a formula $\sigma$ in $\socel\cup\{\strictsq,\strictks\}(\pset)$ equivalent to $\STRICT(\rho)$.
	Thus, $\psi' = \sigma [A \rightarrow B]$ is equivalent to $\STRICT(\varphi')$.
	
	\item If $\varphi' = \rho_1 \sq \rho_2$, then $\STRICT(\varphi') \equiv \STRICT(\rho_1) \strictsq \STRICT(\rho_2)$.
	By induction hypothesis, both $\STRICT(\rho_1)$ and $\STRICT(\rho_2)$ have equivalent formulas $\sigma_1$ and $\sigma_2$, respectively, in $\socel\cup\{\strictsq,\strictks\}(\pset)$.
	Thus, $\psi' = \sigma_1 \strictsq \sigma_2$ is equivalent to $\STRICT(\varphi')$.
	\item If $\varphi' = \rho \FILTER P$, then $\STRICT(\varphi') \equiv \STRICT(\rho) \FILTER P$.
	By induction hypothesis, $\STRICT(\rho)$ has an equivalent formula $\sigma$ in $\socel\cup\{\strictsq,\strictks\}(\pset)$.
	Thus, $\psi' = \sigma \FILTER P$ is equivalent to $\STRICT(\varphi')$.
	\item If $\varphi' = \rho_1 \cor \rho_2$, then $\STRICT(\varphi') \equiv \STRICT(\rho_1) \cor \STRICT(\rho_2)$.
	By induction hypothesis, both $\STRICT(\rho_1)$ and $\STRICT(\rho_2)$ have equivalent formulas $\sigma_1$ and $\sigma_2$, respectively, in $\socel\cup\{\strictsq,\strictks\}(\pset)$.
	Thus, $\psi' = \sigma_1 \cor \sigma_2$ is equivalent to $\STRICT(\varphi')$.
	\item If $\varphi' = \rho \ks$, then $\STRICT(\varphi') \equiv \STRICT(\rho) \strictks$.
	By induction hypothesis, $\STRICT(\rho)$ has an equivalent formula $\sigma$ in $\socel\cup\{\strictsq,\strictks\}(\pset)$.
	Thus, $\psi' = \sigma \strictks$ is equivalent to $\STRICT(\varphi')$.
\end{itemize}
After this the only thing left is to replace every subformula $\varphi'$ of $\varphi$ with its equivalent $\psi'$, and the resulting formula will be in $\socel\cup\{\strictsq,\strictks\}(\pset)$ and will be equivalent to $\varphi$, thus proving the lemma.

\subsection{Proof of Theorem~\ref{theo:strict-proper-containment}}

\subsubsection{\texorpdfstring{$\socel\cup\{:\} \subseteq \socel\cup\{\STRICT\}$}{Lg}}

For the first part we prove that for every formula $\varphi$ in $\socel\cup\{\strictsq\}(\pset)$ there is a formula $\psi$ in $\socel \cup \{\STRICT\}(\pset)$ such that $\varphi\equiv\psi$.
For notation, for any formula $\rho$ and labels $A,B$ we write $\rho^{A\rightarrow B}$ to refer to the formula $\rho$ after replacing every occurrence of $A$ by $B$.

Consider a formula $\varphi$ in $\socel\cup\{\strictsq\}(\pset)$.
We are going to make use of a rather useful property of $\socel\{\strictsq\}(\pset)$, and of $\socel$ in general, that is the fact that one can always push the labelling down to the assignments (or an assignment with filters).
This is because any formula of the form $\varphi = \varphi' \cin A$ can be rewritten as a new formula $\psi$ that labels $A$ one level lower than $\varphi$.
This is done in the following way:
\begin{itemize}
	\item If $\varphi' = \rho_1 ~\OP~ \rho_2$ with $\OP \in \{\sq, \strictsq, \cor\}$, then $\varphi \equiv (\rho_1 \cin A) ~\OP~ (\rho_2 \cin A)$.
	\item If $\varphi' = \rho ~\OP$ with $\OP \in \{\ks, \strictks\}$, then $\varphi \equiv (\rho \cin A) ~\OP$.
	\item If $\varphi' = \rho \FILTER P$, then $\varphi \equiv (\rho \cin A' \FILTER P)[A' \rightarrow A]$, where $A'$ is a new label.
	\item If $\varphi' = \rho [B_1 \rightarrow B_2]$, then $\varphi \equiv \rho^{B_1 \rightarrow A'} \cin A [A' \rightarrow B_2]$, where $A'$ is a new label.
\end{itemize}
By using this equivalences, we can push down all labellings.

Now we prove that for every formula $\varphi' = \varphi_1 \strictsq \varphi_2$ with $\varphi_1$ and $\varphi_2$ in $\socel\cup\{\STRICT\}(\pset)$ there exists a formula $\psi'$ in $\socel\cup\{\STRICT\}(\pset)$ equivalent to $\varphi'$.
Here we assume that $\varphi'$ is such that all its labels are applied at the lower level.
Then the proof follows by doing induction over the length of $\varphi'$.
The base case is $\varphi' = R \strictsq T$ for some $R$, $T$.
Clearly $\varphi'$ is equal to $\psi' = \STRICT(R \sq T)$.
A similar argument works if $\varphi_1$ and $\varphi_2$ have the form $(R \cin A_1 \ldots \cin A_k)[B_1 \rightarrow B_1'] \ldots [B_l \rightarrow B_l']$.
For this case, $\varphi'$ would be equal to $\psi' = \STRICT(\varphi_1 \strictsq \varphi_2)$.
For the inductive step consider the following cases:
\begin{itemize}
	\item If $\varphi_1 = \rho [A \rightarrow B]$, then $\varphi' \equiv (\rho \strictsq (\varphi_2^{A\rightarrow A'}))[A\rightarrow B][A' \rightarrow A]$.
	By induction hypothesis, we know that there is a formula $\sigma$ in $\socel\cup\{\STRICT\}(\pset)$ equivalent to $\rho \strictsq (\varphi_2^{A\rightarrow A'})$.
	Thus, $\psi' = \sigma[A\rightarrow B][A' \rightarrow A]$ is equivalent to $\varphi'$.
	
	\item If $\varphi_1 = \rho_1 \sq \rho_2$, then $\varphi' \equiv \rho_1 \sq (\rho_2 \strictsq \varphi_2)$ and $\rho_2 \strictsq \varphi_2$ is smaller than $\varphi_1 : \varphi_2$.
	By induction hypothesis, $(\rho_2 \strictsq \varphi_2)$ has an equivalent formula $\sigma$ in $\socel\cup\{\STRICT\}(\pset)$.
	Thus, $\psi' = \rho_1 \sq \sigma$ is equivalent to $\varphi'$.
	
	\item If $\varphi_1 = \rho \FILTER P(A_1, \ldots , A_k)$, then $\varphi' \equiv ((\rho \strictsq \varphi_2^\rightarrow) \FILTER P(A_1, \ldots, A_k))^\leftarrow$, where $\psi^\rightarrow = (\psi^{A_1 \rightarrow A_1'}) \ldots ^{A_k \rightarrow A_k'}$ and $\psi^\leftarrow = \psi[A_1' \rightarrow A_1] \ldots [A_k' \rightarrow A_k]$ for some new dummy labels $A_1', \ldots, A_k'$.
	By induction hypothesis, $(\rho \strictsq \varphi_2^\rightarrow)$ has an equivalent formula $\sigma$ in $\socel\cup\{\STRICT\}(\pset)$.
	Thus, $\psi' = (\sigma \FILTER P(A_1, \ldots, A_k))^\leftarrow$ is equivalent to $\varphi'$.
	
	\item If $\varphi_1 = \rho_1 \cor \rho_2$, then $\varphi' \equiv (\rho_1 \strictsq \varphi_2) \cor (\rho_2 \strictsq \varphi_2)$.
	By induction hypothesis, both $(\rho_1 \strictsq \varphi_2)$ and $(\rho_2 \strictsq \varphi_2)$ have equivalent formulas $\sigma_1$ and $\sigma_2$, respectively, in $\socel\cup\{\STRICT\}(\pset)$.
	Thus, $\psi' = \sigma_1 \cor \sigma_2$ is equivalent to $\varphi'$.
	\item If $\varphi_1 = \rho \ks$, then $\varphi' \equiv (\rho \strictsq \varphi_2) \cor (\rho \ks \sq (\rho \strictsq \varphi_2))$.
	By induction hypothesis, $(\rho \strictsq \varphi_2)$ has an equivalent formula $\sigma$ in $\socel\cup\{\STRICT\}(\pset)$.
	Thus, $\psi' = \sigma \cor (\rho \ks \sq \sigma)$ is equivalent to~$\varphi'$.
\end{itemize}
The cases regarding the structure of $\varphi_2$ instead of $\varphi_1$ are analogous to the previous ones.

Then, every subformula $\varphi'$ of $\varphi$ is replaced by its equivalent $\psi'$ in a bottom-up fashion to ensure that the subformulas of $\varphi'$ are indeed in $\socel\cup\{\STRICT\}(\pset)$.
Finally, the resulting formula $\psi$ is in $\socel\cup\{\STRICT\}(\pset)$ and is equivalent to $\varphi$.

\subsubsection{\texorpdfstring{$\socel\cup\{\strictks\} \nsubseteq \socel\cup\{\STRICT\}$}{Lg}}

Now, for the second part we prove that there is a set $\pset$ containing a single binary SO predicate and a formula $\varphi\in\socel\cup\{\strictks\}(\pset)$ that is not equivalent to any formula in $\socel\cup\{\STRICT\}(\pset)$.
In particular, consider $\pset = \{\ext{P_=}\}$, where $P_=(x,y) := (x.a = y.a)$, and consider the formula:
$$
\varphi \; = \; ((A \sq E) \FILTER \ext{P_=}(A,B))\strictks
$$
in $\socel\cup\{\strictks\}(\pset)$.
We prove that there is no formula $\psi$ in $\socel\cup\{\STRICT\}(\pset)$ equivalent to $\varphi$.
For this we first give a few somehow useful definitions. Consider a stream $S$, a complex event $C$, two positions $i,j \in \supp(C)$ with $i < j$ and a constant $k \geq 1$.
Then, the result of pumping the fragment $[i,j]$ of $(S,C)$ $k$-times is a tuple $(S',C')$ where $S'$ and $C'$ are a stream and complex event, respectively, defined as follows.
Consider the factorization $C_1 \cdot C_2 \cdot C_3$ of $C$, where $C_1$ contains all positions in $C$ that are lower than $i$, $C_2$ contains all positions of $C$ between $i$ and $j$ (including them) and $C_3$ contains all positions of $C$ higher than $j$.
Likewise, consider the factorization $S_1 \cdot S_2 \cdot S_3$ of $S$ in the same way.
Now, if $C_2$ does not induce a contiguous interval (that is, if there is some $l$ such that $i < l < j$ and $l \notin \supp(C_2)$), then define $S'$ as $S_1 \cdot P_0 \cdot S_2 \cdot P_1 \cdot S_2 \cdot \ldots \cdot S_2 \cdot P_k \cdot S_3$, where $S_2$ is repeated $k$ times and each $P_i$ is an arbitrary finite stream.
On the other hand, if $C_2$ induces a contiguous interval, define $S'$ the same way but without the $P_i$, i.e., $S' = S_1 \cdot S_2 \cdot S_2 \cdot \ldots \cdot S_2 \cdot S_3$.
Similarly, define $C'$ as $C_1 \cdot C_2^1 \cdot C_2^2 \cdot \ldots \cdot C_2^k \cdot C_3'$, where each $C_2^i$ is the same complex event $C_2$ but with its values moved to fit the $i$-th occurrence of $S_2$ in $S'$.
For example, $C_2^2$ results after adding $|S_2|$ to all values of $C_2(A)$ for every label $A$ if it induces a contiguous interval, and adding $|P_0| + |S_2| + |P_1|$ else.
$C_3'$ is the same as $C_3$ but moved to fit $S_3$.
Notice that if $C$ induces a contiguous interval, then $C'$ also does.
Moreover, notice that no new events were added to the complex event, i.e. $C_S(A) = C'_{S'}(A)$ for every label $A$.

A formula $\rho$ in $\socel$ is said to be \textit{pumpable} if there exists a constant $N \in \bbN$ such that for every stream $S$, positions $p_1,p_2$ and complex event $C \in \sem{\rho}(S,p_1,p_2)$ with $|\supp(C)| > N$ there exist two positions $i,j \in \supp(C)$ with $i < j$ such that for every $k \geq 1$ it holds that $C' \in \sem{\rho}(S',p_1,p_2')$, where $(S',C')$ is the results of pumping the fragment $[i,j]$ of $(S,C)$ $k$-times and $p_2'$ is the position at which $S[p_2]$ ended at.
In the following lemma we show the utility of this property.

\begin{lemma}
	Every formula $\varphi$ in $\socel\cup\{\STRICT\}(\pset)$ is pumpable.
\end{lemma}
\begin{proof}
	Consider a formula $\varphi$ in $\socel\cup\{\STRICT\}(\pset)$.
	We prove the lemma by induction over the length of $\varphi$.
	First, consider the base case $R$.
	Then, by defining $N = 1$ we know that for every stream $S$ there is no complex event $C \in \sem{\varphi}(S)$ with $|\supp(C)| > N$, so the lemma holds.
	
	Now, for the inductive step consider first the case $\varphi = \psi_1 \FILTER P(X_1,\ldots,X_n)$.
	By induction hypothesis, we know that the lemma holds for $\psi_1$, thus let $N_1$ be its corresponding constant.
	Let $N$ be equal to $N_1$.
	Consider any stream $S$, positions $p_1,p_2$ and complex event $C \in \sem{\varphi}(S, p_1,p_2)$ with $|\supp(C)| > N$.
	By definition $C \in \sem{\psi_1}(S, p_1,p_2)$ and $(C_S(X_1),\ldots,C_S(X_n)) \in P$.
	By induction hypothesis, $\psi_1$ is pumpable, thus there exist positions $i,j \in \supp(C)$ with $i < j$ such that the fragment $[i,j]$ can be pumped.
	Moreover, consider that the result of pumping the fragment $[i,j]$  $k$ times is $(S',C')$, for an arbitrary $k$.
	Then, it holds that $C' \in \sem{\psi_1}(S'p_1,p_2')$.
	Also, because in the pumping it holds that $C_S(A) = C'_{S'}(A)$ for every $A$, then $(C'_{S'}(X_1),\ldots,C'_{S'}(X_n)) \in P$.
	Therefore, $C' \in \sem{\varphi}(S'p_1,p_2')$, thus $\varphi$ is pumpable.
	
	Consider now the case $\varphi = \psi_1 \cor \psi_2$.
	By induction hypothesis, we know that the property holds for $\psi_1$ and $\psi_2$, thus let $N_1$ and $N_2$ be the corresponding constants, respectively.
	Then, we define the constant $N$ as the maximum between $N_1$ and $N_2$.
	Consider any stream $S$, positions $p_1,p_2$ and complex event $C \in \sem{\varphi}(S,p_1,p_2)$ with $|\supp(C)| > N$.
	Either $C \in \sem{\psi_1}(S,p_1,p_2)$ or $C \in \sem{\psi_2}(S,p_1,p_2)$, so w.l.o.g. consider the former case.
	By induction hypothesis, $\psi_1$ is pumpable, thus there exist positions $i,j \in \supp(C)$ with $i < j$ such that the fragment $[i,j]$ can be pumped and the result $(C',S')$ satisfies $C' \in \sem{\psi_1}(S',p_1,p_2')$.
	This means that $C' \in \sem{\varphi}(S',p_1,p_2')$, therefore, $\varphi$ is pumpable.
	
	Now, consider the case $\varphi = \psi_1 \sq \psi_2$.
	By induction hypothesis, we know that the property holds for $\psi_1$ and $\psi_2$, thus let $N_1$ and $N_2$ be the corresponding constants, respectively.
	Then, we define the constant $N = N_1 + N_2$.
	Consider any stream $S$, positions $p_1,p_2$ and complex event $C \in \sem{\varphi}(S,p_1,p_2)$ with $|\supp(C)| > N$.
	This means that there exist complex events $C_1$ and $C_2$ with $C = C_1 \cdot C_2$ such that $C_1 \in \sem{\psi_1}(S,p_1,p')$ and $C_2 \in \sem{\psi_2}(S,p'+1,p_2)$, where $p' = \max(M_1)$.
	Moreover, either $|\supp(C_1)| > N_1$ or $|\supp(C_2)| > N_2$, so w.l.o.g. assume the former case.
	By induction hypothesis, $\psi_1$ is pumpable, thus there exist positions $i,j \in \supp(C_1)$ with $i < j$ such that the fragment $[i,j]$ can be pumped and the result $(C_1',S')$ satisfies $C_1' \in \sem{\psi_1}(S',p_1,p'+r)$, assuming that the pumping added $r$ new events.
	Define the complex event $C' = C_1' \cdot C_2'$, where $C_2'$ is the same as $C_2$ but adding $r$ to each position (so that $(C_2)_S = (C_2')_{S'}$).
	Then $C_1' \in \sem{\psi_1}(S',p_1,p'+r)$ and $C_2' \in \sem{\psi_2}(S,p'+r+1,p_2+r)$, thus $C' \in \sem{\varphi}(S',p_1,p_2+r)$, therefore, $\varphi$ is pumpable.
	
	Consider then the case $\varphi = \psi_1 \ks$.
	By induction hypothesis, we know that the lemma holds for $\psi_1$, thus let $N_1$ be its corresponding constant.
	Let the constant $N$ be equal to $N_1$.
	Consider any stream $S$, positions $p_1,p_2$ and complex event $C \in \sem{\varphi}(S,p_1,p_2)$ with $|\supp(C)| > N$.
	Then, consider $i = \min(C)$ and $j = \max(C)$, consider any $k \geq 1$ and let $(S',C')$ be the result of pumping the fragment $[i,j]$ of $(S,C)$ $k$ times.
	We prove now that $C' \in \sem{\varphi}(S',p_1,p_2')$ by induction over $k$.
	If $k = 1$ then, as defined in the definition of pumping, $S'$ has the form $S_1 \cdot P_0 \cdot S_2 \cdot P_1 \cdot S_3$, and $C'$ is the same as $C$ but adding $r$ to each position, where $r = |P_0|$.
	Clearly it holds that $C' \in \sem{\varphi}(S',p_1,p_2')$, since the modifications did not affect the complex event part of $S$.
	Now, consider that $k > 1$.
	Then, $S'$ has the form $S_1 \cdot P_0 \cdot S_2 \cdot P_1 \cdot S_2 \cdot \ldots \cdot S_2 \cdot P_k \cdot S_3$.
	Similarly, $C'$ is defined as $C_1 \cdot C_2^1 \cdot C_2^2 \cdot \ldots \cdot C_2^k \cdot M_3'$, where $C_1 = C'_3 = \emptyset$ and each $C_2^i$ is the same complex event $C$ but with its positions moved to fit the $i$-th occurrence of $S_2$ in $S'$.
	Consider that $r = |S_1 \cdot P_0 \cdot S_2|$.
	By induction hypothesis, we can say that $C_2' \in \sem{\varphi}(S',r+1,p_2')$ where $C_2' = C_2^2 \cdot \ldots \cdot C_2^k$ (notice that we consider it from position $r+1$ because there is no lower position in the complex event).
	Also, it is easy to see that this implies $C_2' \in \sem{\varphi \ks}(S',r,p_2')$, which is something we will need next.
	Moreover, it holds that $C_2^1 \in \sem{\varphi}(S',p_1,r)$, because it represents the same complex event as the original one $C$.
	Then, because $C' = C_2^1 \cdot C_2'$, it follows that $C' \in \sem{\varphi \sq \varphi \ks}(S',p_1,p_2')$ which also implies that $C' \in \sem{\varphi \ks}(S',p_1,p_2')$.
	Since $\varphi \ks = \psi_1 \ks \ks \equiv \psi_1 \ks = \varphi$, it holds that $C' \in \sem{\varphi}(S',p_1,p_2')$.
	
	Now, consider the case $\varphi = \STRICT(\psi_1)$.
	By induction hypothesis, we know that the lemma holds for $\psi_1$, thus let $N_1$ be its corresponding constant.
	Let the constant $N$ be equal to $N_1$.
	Consider any stream $S$, positions $p_1,p_2$ and complex event $C \in \sem{\varphi}(S,p_1,p_2)$ with $|\supp(C)| > N$.
	Then, by definition $C \in \sem{\psi_1}(S,p_1,p_2)$, and by induction hypothesis there exist positions $i,j \in \supp(C)$ such that the fragment $[i,j]$ can be pumped and the result $(S',C')$ satisfies $C' \in \sem{\psi_1}(S,p_1,p_2')$.
	Notice that $C$ induces a contiguous interval because of the definition of $\STRICT$, therefore $C'$ also induces a contiguous interval, thus $C' \in \sem{\varphi}(S,p_1,p_2')$.
	
	Finally, consider the case $\varphi = \psi_1[A \rightarrow B]$.
	By induction hypothesis, we know that the lemma holds for $\psi_1$, thus let $N_1$ be its corresponding constant.
	Let the constant $N$ be equal to $N_1$.
	Consider any stream $S$, positions $p_1,p_2$ and complex event $C \in \sem{\varphi}(S,p_1,p_2)$ with $|\supp(C)| > N$.
	Then, by definition there exists $D \in \sem{\psi_1}(S,p_1,p_2)$ such that $C(B) = D(A) \cup D(B)$ and $C(A)= \emptyset$.
	By induction hypothesis there exist positions $i,j \in \supp(D)$ such that the fragment $[i,j]$ can be pumped and the result $(S',D')$ satisfies $D' \in \sem{\psi_1}(S,p_1,p_2')$.
	Since $\supp(C) = \supp(D)$, one can see that pumping $[i,j]$ in $(S,C)$ is the same as pumping $[i,j]$ in $(S,D)$ with the only difference that the results $(S',C')$ and $(S,D')$ satisfy $C'(B) = D'(A) \cup D'(B)$ and $C'(A) = \emptyset$.
	Then it follows that $C' \in \sem{\varphi}(S,p_1,p_2')$.
	A similar argument works for the case $\varphi = \psi_1 \cin A$.
\end{proof}
Now, we show that there is no formula $\psi$ in $\socel\cup\{\STRICT\}(\pset)$ equivalent to $\varphi$ by proving that such formula is not pumpable.
By contradiction, assume that $\psi$ exists, and let $N$ be its constant.
Consider then the stream:
\[
S = \begin{array}{cccc}
\begin{array}{cccccc}
A & L & E & A & L & E \\
1 & 1 & 1 & 2 & 2 & 2
\end{array} & \cdots & \begin{array}{ccc}
A & L & E \\
N & N & N
\end{array} & \cdots
\end{array}
\]
Where the first and second lines are the type and $a$ attribute of each event, respectively, and consider the complex event $C$ with $C(A) = \{1, 4, 7, \ldots , 3N - 2\}$ and $C(E) = \{3, 6, 9, 3N\}$.
Now, let $i,j \in \supp(C)$ be any two positions of the complex event, which define the partition $C_1 \cdot C_2 \cdot C_3$, and name $t_1 = S[i]$ and $t_2 = S[j]$.
We will use $k = 2$, i.e., repeat section $S[i,j]$ two times, and use the 1-tuple stream $U(0)$ as the arbitrary streams $P_0$ and $P_1$ to get the resulting stream $S'$ and the corresponding complex event $C'$.
We will analyse the following possible cases: $\type(t_1) = \type(t_2)$; $\type(t_1) = A$ and $\type(t_2) = E$; $\type(t_1) = E$ and $\type(t_2) = A$.
In the first case the resulting $C'$ is a complex event with two consecutive tuples of the same type, which contradicts the original formula $\varphi$.
In the second case $C_2$ is not a contiguous interval so the complex event $C'$ would fail to ensure that the $A$ tuple following $t_2$ is placed right after it (because of the tuple $U(0)$ between), thus contradicting the $\strictks$ property of $\varphi$.
In the third case it is clear that the last $A$ in the first repetition of $[i,j]$ and the first $E$ in the second repetition (i.e., $S[j]$ and $S[j+2]$) do not satisfy the $\sFILTER$ condition because $S[j].a > S[j+2].a$.
Finally, the formula $\psi$ cannot exist.


\end{document}